\documentclass[11pt]{article}
\usepackage{graphicx, fullpage, amsmath, amssymb} 
\usepackage{setspace}
\usepackage{natbib}
\usepackage{doi}

\usepackage{geometry}
\usepackage{microtype}
\usepackage{graphicx}
\usepackage{subfigure}
\usepackage{booktabs} 
\usepackage{algorithmic}
\usepackage[ruled,vlined]{algorithm2e}
\usepackage{balance}
\usepackage{hyperref}

\usepackage{mathtools}
\usepackage{bm}
\usepackage{amsthm, tikz, nicefrac}

\usepackage[capitalize,noabbrev]{cleveref}

\theoremstyle{plain}
\newtheorem{theorem}{Theorem}[section]

\newtheorem{lemma}[theorem]{Lemma}
\newtheorem{corollary}[theorem]{Corollary}
\theoremstyle{definition}

\theoremstyle{remark}

\newtheorem*{example*}{Example}

\newcommand{\R}{\mathbb{R}}
\newcommand{\E}{\mathbb{E}}
\renewcommand{\P}{\mathbb{P}}

\newcommand{\WPB}{\textsf{WPB}}

\newcommand{\AP}{\textsf{AP}}

\newcommand{\abne}{\mathcal{A}}
\newcommand{\pbne}{\mathcal{P}}
\newcommand{\babne}{\bm{\mathcal{A}}}
\newcommand{\bpbne}{\bm{\mathcal{P}}}
\newcommand{\abic}{{A}}
\newcommand{\pbic}{{P}}
\newcommand{\babic}{\bm{{A}}}
\newcommand{\bpbic}{\bm{{P}}}
\newcommand{\asc}{A^{\text{SC}}}
\newcommand{\basc}{\bm{A}^{\text{SC}}}
\newcommand{\bv}{\bm{v}}
\newcommand{\bb}{\bm{b}}

\newcommand{\one}{\mathbf{1}}

\newcommand{\Var}{\mathsf{Var}}
\newcommand{\Rev}{\mathsf{Rev}}

\date{}
\title{Revenue Variance Minimization: \\
Beyond First Price Auctions}

\author{
  \begin{tabular}{c@{\hspace{0.5in}}c@{\hspace{0.5in}}c}
    Marek Bojko\thanks{Department of Economics, Yale University, marek.bojko@yale.edu} & 
    Preston McAfee\thanks{Google Research, mcaf@google.com} & 
    Renato Paes Leme\thanks{Google Research, renatoppl@google.com} \\[12pt]
    \multicolumn{3}{c}{
      Balasubramanian Sivan\thanks{Google Research, balusivan@google.com} \qquad\qquad
      Sergei Vassilvitskii\thanks{Google Research, sergeiv@google.com}
    }
  \end{tabular}
}

\begin{document}

\maketitle

\begin{abstract}
We study revenue variance in the sale of $k$ homogeneous items to risk-neutral, unit-demand bidders with independent private values. Although the Revenue Equivalence Theorem implies that standard auctions generate the same expected revenue, the distribution of revenue differs across mechanisms. Prior work shows that, in single-item environments with ex-post individual rationality (IR), the first-price auction minimizes revenue variance. We show that this result is fragile. Under interim IR, the optimality of the first-price auction breaks down in asymmetric single-item settings, and we characterize the variance-minimizing mechanisms for any implementable allocation rule in this environment. In multi-item symmetric regular environments with interim IR, we construct a mechanism that implements the efficient allocation and guarantees constant revenue while maintaining non-negative payments. Under ex-post IR, we show that revenue variance can be reduced relative to winner-pays-bid formats by introducing negative correlations in payments. Nevertheless, we show that the variance ranking between the winner-pays-bid auction and the uniform $(k+1)$-st price auction is maintained in multi-unit settings.
\\
\vspace{0in}
\vspace{0in}
\end{abstract}

\newpage

\section{Introduction}

\subsection{Motivation}

Auctions are both empirically and theoretically important. They allocate a vast range of goods and services, from financial assets, natural resources, and spectrum licenses to internet advertising, real estate, art, and corporate control, and are also widely used in procurement, where governments purchase everything from infrastructure to everyday supplies. Theoretically, auctions serve as a canonical model of price formation and information aggregation, and advances in auction theory have led to influential market designs that govern hundreds of billions of dollars of economic activity.

There are four canonical auction formats \citep{MilgromWeber1982}. In the English auction, familiar from auction houses such as Sotheby’s, bidders are present in a (virtual) room, a provisional winner is identified at each price, and the price is successively increased until no bidder is willing to pay more. This format is commonly used for art and antiques, agricultural commodities, estate sales, and residential real estate in Australia, and is often what people have in mind when referring to an auction. The ``pay-as-bid,'' ``winner-pays-bid,'' or ``first-price sealed-bid'' auction is widely used by governments to sell oil leases and to procure goods and services. In this format, bidders independently submit sealed bids, the highest bidder wins, and pays her bid; no feedback is provided until the auction concludes. The Dutch auction, traditionally used in tulip markets, operates as a reverse English auction: the price starts high and is successively lowered until a bidder accepts it. Finally, the second-price sealed-bid, or Vickrey, auction is similar to the first-price auction in that bids are sealed and the highest bidder wins, but the price paid is the second-highest bid rather than the highest bid. Awareness of the Vickrey auction in the economics literature dates back to \citet{Vickrey1961}, although stamp dealers employed this format in mail-based auctions as early as 1893 \citep{Lucking-Reiley2000}. The Vickrey mechanism is also used in eBay auctions and in some internet advertising auctions, where the winner pays the second-highest bid (plus a small increment). The appeal of this format for stamp dealers was that it replicated the outcome of an English auction without requiring bidders to be present simultaneously.

Under the standard auction-theoretic framework with independent private values that are identically distributed across bidders, all four auction formats yield the same expected revenue to the seller, a consequence of the Revenue Equivalence Theorem \citep{Holt1980,HarrisRaviv1981,Myerson1981,RileySamuelson1981,Vickrey1961,MilgromWeber1982}. Thus, within the canonical model, there is no revenue-based rationale for the seller to prefer one auction format over another.

Nonetheless, expected revenue is only one moment of the revenue distribution. Auctioneers, ranging from governments to digital platforms, often operate under strict budget and regulatory constraints where revenue risk carries tangible economic costs. For governments, high variance jeopardizes the funding of public expenditures. Similarly, firms managing repeated auctions, such as those for advertising or cloud computing, may seek to minimize cashflow volatility.

Canonical auction formats may differ significantly in the distribution of revenue. Thus, even though all of these formats achieve the same expected revenue, the ex-post payment rule becomes a tool for risk management.

\citet{waehrer1998auction} establish that for single-item auctions under ex-post participation constraints, among efficient transaction mechanisms, winner-pays-bid format minimize revenue variance and, more generally, any convex risk measure. In particular, they minimize the variance of the seller’s revenue and profits, as well as the variance of buyers’ profits. As a result, a risk-averse seller may prefer the winner-pays-bid auction to alternative formats, potentially justifying their empirical prevalence.\footnote{Other justifications for the use of winner-pays-bid formats have been suggested in the literature, including reduced scope for within-auction punishment that can sustain collusion in open formats \citep{Milgrom1987} and the absence of certain pathological equilibria that can arise in second-price auctions. However, these considerations do not appear decisive empirically, as collusion has also been documented in repeated winner-pays-bid settings \citep{Pesendorfer2000}. Other reasons include credibility \citep{akbarpour2020credible}. On the other hand, winner-pays-bid may potentially introduce new vulnerabilities such as bid leakage or favoritism by the auctioneer in repeated environments \citep{BergemannHorner2017}.}

\subsection{Our Contribution} 

In this paper, we investigate the robustness of the result that the winner-pays-bid format minimizes revenue variance. We examine two dimensions along which this result fails. First, we relax the participation constraint from ex-post to interim individual rationality (IR) and consider asymmetric environments. Second, we extend the analysis from single-item settings to auctions that allocate multiple items.

\paragraph{Beyond Ex-Post IR.} The ex-post IR constraint restricts attention to auction formats in which losers make no payments. In many real-world settings, however, losing bidders incur substantial participation costs.  Bidders must evaluate the item, whether it is of good quality and meets their needs.  These costs are not recoverable if the bidder loses the auction.  Bidders in off-shore oil auctions \cite{hendricks1993bidding} must estimate not just the amount of oil available, but the costs of getting that oil to their refineries.  About half of the offered tracts get no bids at all.  In government procurement, a company 
must estimate the cost of providing the goods or services, and just estimating these costs can be substantial~\cite{mcafee1988incentives}. The design of crowdsourcing contests \cite{chawla2019optimal} is an example of an auction format where bids are not monetary payments, but participants exert efforts, and hence a form of all-pay auctions. Thus, in many situations, an assumption that losing bidders pay nothing is unreasonable.  While these situations also lead to a situation with an endogenous number of bidders \cite{mcafee1987auctions,mcafee2002set,waehrer1998auction}, provided the set of potential bidders is not too large, all will enter, bear the costs, and thus have negative profits conditional on losing.  Indeed, contests can be viewed as all-pay auctions \cite{fullerton1999auctionin}, where a bid is effort rather than cash.  In many real world situations, `losers pay nothing' is an inappropriate assumption.

In light of this discussion, we ask whether winner-pays-bid auctions remain risk-minimal when losing bidders may incur payments, while interim individual rationality is maintained, as in all-pay auctions. We first show that in a setting with independently and identically (i.i.d.) distributed values, the winner-pays-bid mechanism minimizes any convex risk measure of revenue for any implementable allocation rule (Theorem \ref{thm:symmetric}). Establishing this result requires an argument distinct from that of \citet{waehrer1998auction}, whose analysis relies critically on ex-post individual rationality. We develop a new technique that decomposes the revenue function into two components: one that depends only on the type of a fixed bidder $i$, and a second that satisfies a complementary slackness property, namely that it vanishes whenever bidder $i$ makes a nonzero payment.

More interestingly, the same decomposition shows that winner-pays-bid auctions need not be variance-minimal in non-i.i.d.\ environments. We use this approach to characterize the optimal auction format in such settings (Theorem~\ref{thm:asymmetric}). Allowing payments from losing bidders, while preserving interim individual rationality, enables transfers to be spread across states in a way that reduces the seller’s revenue risk. To our knowledge, this yields a novel auction format.

We further show that for any value distribution and implementable allocation rule, there exists a payment rule satisfying these optimality conditions (Theorem~\ref{theorem: existence pi}). The proof uses a fixed-point argument that may be of independent interest. Rather than constructing payment functions directly, we reformulate the problem in terms of the probability measures of the random variables induced by the conjectured payment scheme and bidders' values. We then define a continuous operator on this space and establish existence of a fixed point via the Schauder fixed-point theorem. Finally, we define payment functions based on the obtained fixed point.

\paragraph{Beyond Single-Item Auctions.} Another crucial assumption underlying the analysis of \citet{waehrer1998auction} is that only a single item is allocated and the allocation is deterministic. In such settings, for any realization of types, at most one agent makes a nonzero payment. Once multiple agents make nonzero payments, however, the auctioneer can introduce negative correlation across payments, thereby reducing revenue variance. We show that this can be achieved even under ex-post individual rationality, allowing revenue variance to fall below the level attained by discriminatory winner-pays-bid formats. 

Under interim individual rationality, we construct a mechanism that implements the efficient allocation and delivers constant revenue while maintaining non-negative payments in i.i.d.\ regular environments (Theorem~\ref{thm:general_n_k}). The optimal payment rule is characterized as the solution to a differential equation. The main technical challenge is to show that this solution remains within a bounded interval. Establishing this requires several nontrivial applications of properties of \citet{Myerson1981}'s virtual value function to analyze the critical points of the differential equation.

Nonetheless, while not variance-minimal, we show in Theorem \ref{th:varcomp} that the discriminatory winner-pays-bid auction guarantees lower revenue variance than the uniform $(k\!+\!1)$st-price auction, extending the corresponding result from single-item settings.

\subsection{Related Work}
Beyond the paper of \citet{waehrer1998auction}, \citet{eso1999auction} also study the problem of minimizing risk in settings with independent private value settings. The latter paper drops the requirement that payments need to be non-negative and, with that, is able to produce payment rules with zero variance. We propose an efficient mechanism for symmetric multi-item regular environments that achieves zero revenue variance while satisfying interim individual rationality and maintaining non-negative payments.

Another important line of work is undertaken by \citet{SundararajanYan} and \citet{bhalgat2012mechanism}, who design allocation and payment rules to minimize the seller's risk. A conclusion of their work is that the optimal design is highly dependent on the risk measure used by the seller. Instead, we fix the allocation rule and optimize among the payment rules that implement that allocation. Remarkably, the optimal design in this setting is independent of the risk measure used by the seller. 

\citet{pekevc2025variance} show that, in multi-item settings, the winner-pays-bid auction yields lower revenue variance than the uniform auction under log-concavity of bidders’ valuation distributions. We show that this result extends to general environments, dispensing with the log-concavity assumption.



Not surprisingly, risk aversion has been examined in the literature, but as far as we can tell, only from the perspective of a risk-neutral seller selling to risk-averse buyers. \citet{HarrisRaviv1981}, \citet{Holt1980}, \citet{MaskinRiley1984}, and \citet{RileySamuelson1981} all observe that, under independent private values but risk-averse bidders, the winner-pays-bid auction produces a higher revenue than the English auction.  There is a compelling intuition underlying this result.  Risk aversion on the part of the bidders does not affect the outcome in the English auction, which still ends at the second-highest value.  Meanwhile  winner-pays-bid presents a bidder with risk, in particular the risk of losing.  Bidders can reduce that risk by increasing their bid, leveling the outcomes slightly while increasing the likelihood of the better outcome, which is desirable for risk-averse bidders.  Given the Revenue Equivalence Theorem, a bit of risk aversion favors the winner-pays-bid auction.

In addition, the literature also considered optimal auctions for risk-averse bidders.  While optimal auctions with risk-averse bidders are immensely complicated, and even today we do not have the mathematical tools to evaluate randomized mechanisms offered to risk-averse bidders, \citet{MaskinRiley1984}, \citet{Matthews1983} and \citet{Moore1984} show that optimal auctions involve subsidizing high bidders who lose and penalizing low bidders, using risk to encourage higher bids.  More recently, \citet{Chawla18} made progress on this problem, giving a characterization that leads to approximately optimal results.  However, none of these papers considers the case of risk-neutral bidders and a risk-averse seller, which is the setting we consider in this work.

\section{Model and Preliminary Results}

\subsection{Setup and Auction Mechanisms}

\paragraph{Setup.} We consider a setting with $k$ homogeneous items and $n > k$ risk-neutral agents with unit demand. Each agent $i$ derives a private value $v_i \in \mathcal{V}_i \equiv [0, \overline{v}_i] \subseteq \mathbb{R}_{+}$ from the item. We define the product space $\mathcal{V} = \prod_i \mathcal{V}_i$ and denote a value profile by $\bv \in \mathcal{V}$. The value $v_i$ is agent $i$'s private information and we assume it is drawn according to a full-support prior $F_i$ with density $f_i > 0$. We assume values across agents are statistically independent. Some of our results also require that the bidders' values are drawn i.i.d., and we mention those where applicable. In this case, we drop agent indices on the objects introduced above.


\paragraph{Auctions.} An \emph{auction} is a mechanism that elicits a bid $b_i \in [0,\infty)$ from each bidder $i$, allocates the items for sale, and charges based on the allocation and the submitted bids. For each bid vector $\bb = (b_1, \hdots, b_n)$ the auction returns a random variable $\babne(\bb) \in \{0,1\}^n$ that specifies which bidder gets the item and a random variable $\bpbne(\bb) \in [0,\infty)^n$ that specifies the payments. The allocation rule satisfies $\sum_i \abne_i(\bb) \leq k$ and $\abne_i(\bb) \leq 1$ for each agent $i$.

A \emph{bidding strategy} for agent $i$ is a mapping $b_i : \mathcal{V}_i \rightarrow [0, \infty)$, which maps their value $v_i$ to a bid $b_i$. We refer to $\bb(\bv) = (b_1(v_1), \hdots, b_n(v_n))$. We say that a set of bidding strategies forms a \emph{Bayes-Nash equilibrium} (BNE) whenever:
$$\E[v_i \abne_i(\bb(\bv)) - \pbne_i(\bb(\bv)) \mid v_i] \geq \E[v_i \abne_i(b'_i, \bb_{-i}(\bv)) - \pbne_i(b'_i, \bb_{-i}(\bv)) \mid v_i], \qquad \forall i, v_i, b'_i.$$

We say that an auction $(\babne(\bb), \bpbne(\bb))$ together with equilibrium bidding function $\bb(\bv)$ \emph{implement} social choice function $\basc : [0,\infty)^n \rightarrow [0,1]^n$ if:
$$\basc(\bv) = \E[\babne(\bb(\bv)) \mid \bv].$$
The expectation above is over the randomness in the allocation function. In other words, the social choice function $\basc$ specifies the expected allocation given the values $\bv$, while $\babne$ is a random variable that specifies which bidder gets the item given the bids $\bb$. We also define the interim social choice function  $x_i : \mathcal{V}_i \rightarrow [0,1]$ as:
$$x_i(v_i) = \E[\asc_i(\bv) \mid v_i].$$
A social choice function is \emph{implementable} if there is an auction and an equilibrium bidding function implementing it.  A well-known result by \citet{Myerson1981} states that a social choice function is implementable if and only if its corresponding interim social choice functions are monotone non-decreasing.

An important observation is that the same social choice function can be implemented by many different auction formats. Consider, for example, the efficient allocation function which allocates to the bidder with the highest value. This social choice function can be implemented by (among others) the second-price auction (SP), the first-price auction (FP) and the all-pay auction (AP). In all of these formats: $\abne_i(\bb) = 1$ whenever $b_i$ is the highest bidder (breaking ties lexicographically) and $\abne_i(\bb) = 0$ otherwise. The payment rules are:
\begin{equation}\label{eq:payment_rule_bne}
\pbne_i^{\textsf{SP}}(\bb) = \abne_i(\bb) \max_{j \neq i} b_j \qquad \qquad \pbne_i^{\textsf{FP}}(\bb) = \abne_i(\bb) b_i \qquad \qquad \pbne_i^{\textsf{AP}}(\bb) = b_i
\end{equation}
Each of these auctions induces a different bidding function. Assume, for example, we have two bidders with i.i.d. uniform $[0,1]$ values. Then the bidding functions under BNE in each auction format are:
\begin{equation}\label{eq:bidding_function_bne}
b_i^{\textsf{SP}}(v_i) = v_i  \qquad \qquad
b_i^{\textsf{FP}}(v_i) = v_i/2 \qquad \qquad
b_i^{\textsf{AP}}(v_i) = v_i^2/2
\end{equation}

\paragraph{Direct Revelation Mechanisms.} By the revelation principle \citep{Myerson1981}, given an auction $(\babne(\bb), \bpbne(\bb))$ and a Bayesian Nash equilibrium $\bb(\bv)$, it is possible to construct an alternative auction defined by
$$
\babic(\bv) = \babne(\bb(\bv)) \qquad \bpbic(\bv) = \bpbne(\bb(\bv)).
$$
This auction induces exactly the same distribution of allocations and payments and admits truthful bidding as a Bayesian Nash equilibrium. Intuitively, the auction performs the equilibrium bid shading on the bidders’ behalf, so bidders optimally submit their true values. Since truthful bidding constitutes a Bayesian Nash equilibrium, we say that this auction format is \emph{Bayesian incentive compatible} (BIC). Formally, an allocation rule $\babic(\bv)$ and payment rule $\bpbic(\bv)$ are BIC if
\begin{equation}\label{eq:bic}\tag{BIC}
\E[v_i \abic_i(\bv) - \pbic_i(\bv) \mid v_i] 
\ge 
\E[v_i \abic_i(v'_i, \bv_{-i}) - \pbic_i(v'_i, \bv_{-i}) \mid v_i],
\quad \forall i, v_i, v'_i .
\end{equation}

Continuing our example of two bidders with i.i.d.\ uniform valuations, substituting the equilibrium bidding functions from~\eqref{eq:bidding_function_bne} into~\eqref{eq:payment_rule_bne}, the corresponding payment rules that implement the auctions in BIC form are
\begin{equation}\label{eq:example_payment_rules}
\pbic_i^{\textsf{SP}}(\bv) = \abne_i(\bv)\max_{j \neq i} v_j 
\qquad 
\pbic_i^{\textsf{FP}}(\bv) = \abne_i(\bv/2)\, v_i/2 = \abic_i(\bv)\, v_i/2 
\qquad 
\pbic_i^{\textsf{AP}}(\bv) = v_i^2/2 .
\end{equation}
We refer to these as the BIC implementations of the SP, FP, and AP auctions, respectively.


\subsection{Revenue Equivalence Theorem} The celebrated Revenue Equivalence Theorem \citep{Myerson1981,Milgrom1987} says that in independent private value settings, the expected payment of an agent conditioned on their type depends only on the interim social choice function. More precisely, if $\bb(\bv)$ is a BNE of auction $\babne(\bb), \bpbne(\bb)$ and if $x_i(v_i)$ is the corresponding interim allocation function then:
\begin{equation}\label{eq:ret}\tag{RET}
\E[\pbne_i(\bb(\bv)) \mid v_i] = z_i(v_i) := v_i x_i(v_i) - \int_0^{v_i} x_i(u) du
\end{equation}
In particular, the expected total revenue $\E[\sum_i \pbne_i(\bb(\bv))]$ of SP, FP and AP auctions is exactly the same in i.i.d. environments since the social choice function implemented is identical. In the example with two uniform $[0,1]$ i.i.d. bidders, all of the auctions have expected revenue $\nicefrac{1}{3}$ but, importantly, the distribution of the revenue is different (see Figure \ref{fig:revenue_distribution}).

\begin{figure}[h]
\begin{center}
\includegraphics[scale=.45]{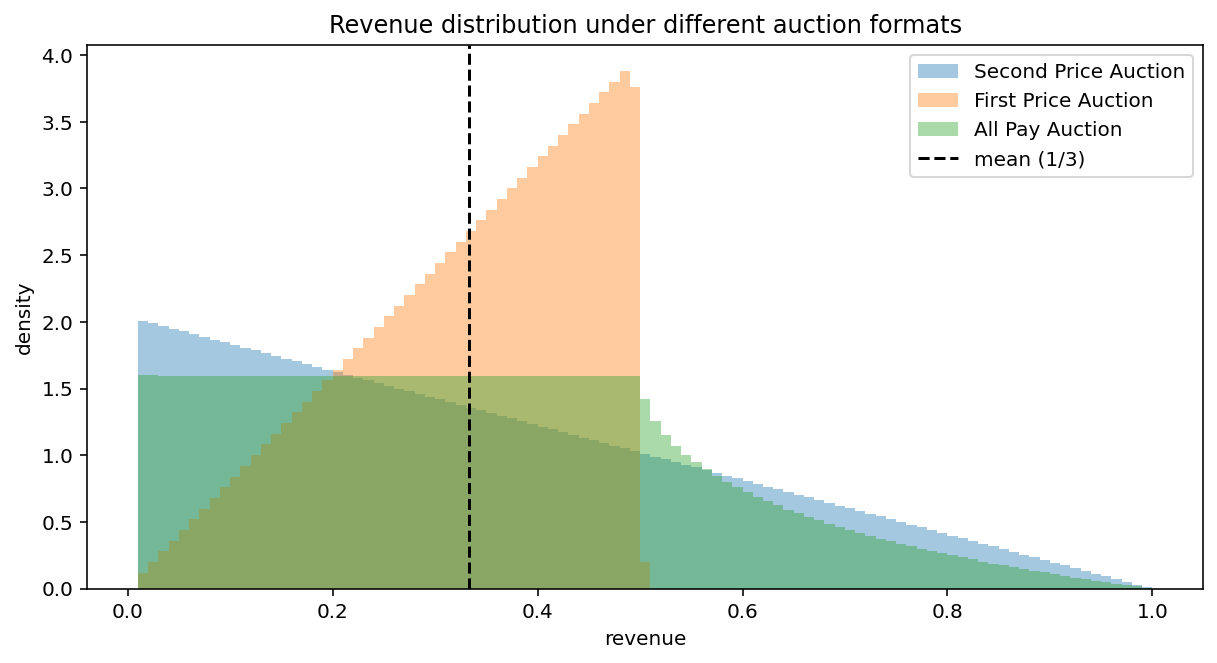}
\end{center}
\caption{The revenue distribution of the second-price, first-price, and all-pay auctions in case of two bidders with i.i.d. Uniform [0,1] value distributions. Observe that while all three formats give the same expected revenue of $\nicefrac{1}{3}$, the underlying distribution of revenue varies greatly across them.}
\label{fig:revenue_distribution}
\end{figure}

\subsection{Variance and Risk Minimization} 

Given that different auction formats can implement the same social choice function, which auction format minimizes the variance of revenue? When the social choice function is fixed, expected revenue is pinned down by the Revenue Equivalence Theorem; however, different auctions can induce markedly different revenue distributions, as illustrated in Figure~\ref{fig:revenue_distribution}.  If $R =\sum_i \pbne_i(\bb(\bv))$ is a random variable representing the revenue, we are interested in minimizing $\textsf{Var}(R)= \E[(R-\E[R])^2] = \E[R^2] - \E[R]^2$. Since $\E[R]^2$ is the same under all auction formats, we are interested in minimizing the second moment $\E[R^2]$.

More generally, we will be interested in minimizing $\E[g(R)]$ for a convex function $g:\R \rightarrow \R$. We will be also interested in quantifying the difference in distributions of other quantities such as $\E[g(U_i)]$, $\E[g(U)]$ and $\E[g(P_i)]$, where $U_i = v_i \cdot \abne_i(\bb(\bv)) - \pbne_i(\bb(\bv))$ is the utility of bidder $i$, $U = \sum_i U_i$ is the sum of utilities of all bidders, and $P_i = \pbne_i(\bb(\bv))$ is the payment of bidder $i$.

\subsection{Winner-Pays-Bid Payment Rule} At the center of this paper is the notion of the \emph{winner-pays-bid} (WPB) payment rule. 
Given an implementable allocation function $\babic(\bv)$, let $x_i(v_i)$ be its interim allocation and $z_i(v_i)$ be the interim payment rules defined in equation \eqref{eq:ret}. We define the WPB payment rule as:
\begin{equation}\label{eq:wpb}\tag{WPB}
\pbic_i^{\textsf{WPB}}(\bv) = b_i^{\textsf{WPB}}(v_i)\abic_i(\bv)\text{, where } b_i^{\textsf{WPB}}(v_i) := \frac{z_i(v_i)}{x_i(v_i)}
\end{equation}
where $b_i^{\textsf{WPB}}(v_i)$ refers to what the bid \emph{would have been} in a BNE implementation of the allocation function $\babic(\cdot)$, where the winner pays their bid, and other agents pay zero. Note however that $\pbic_i^{\textsf{WPB}}(\bv)$ itself results in a BIC implementation of the same allocation function $\babic(\cdot)$.\footnote{In other words, if we are given an allocation function $\babic(\cdot)$ and are asked to implement it as winner-pays-bid, we have two choices.
\begin{enumerate}
\item In a BNE implementation, the agents submit a bid of $b_i^{\textsf{WPB}}(v_i)$. We invert the bid using the inverse of $b_i^{\textsf{WPB}}(\cdot)$ to get the true value $v_i$. We implement the allocation function $\babic(\bv)$ on the true values $\bv$ obtained from inversion, and charge agent $i$ a payment of $b_i^{\textsf{WPB}}(v_i)\abic_i(\bv)$.
\item In a BIC implementation, the agents submit their true values, and we implement the allocation function $\babic(\bv)$ directly on the submitted true values, and charge them a payment of $\pbic_i^{\textsf{WPB}}(\bv) = b_i^{\textsf{WPB}}(v_i)\abic_i(\bv)$, where $b_i^{\textsf{WPB}}(v_i)$ evaluates to  $\frac{z_i(v_i)}{x_i(v_i)}$.
\end{enumerate}
We find the BIC version more convenient to write about because it doesn't involve the inversion of bids to get values, even though for every valuation profile $\bv$, there is no difference between these two versions in the distribution of outcomes and payments they produce. Hereafter, we only focus on the BIC versions $\babic(\cdot)$ and $\bpbic(\cdot)$.}

The astute reader might have observed that for the winner-pays-bid payment rule to make sense, the function $\frac{z_i(v_i)}{x_i(v_i)}$ has to be monotonic non-decreasing. Indeed, this turns out to be the case, as formalized in the Lemma below.

\begin{lemma}\label{lemma:symmetric1} For any implementable allocation function $\babic(\bv)$ and corresponding interim allocation $x_i(\bv)$ and interim payments $z_i(\bv)$, the function $b_i^{\WPB}(v_i) := z_i(v_i) / x_i(v_i)$ is monotone non-decreasing.
\end{lemma}

\begin{proof}
We differentiate the expression for $b_i^{\WPB}(v_i)$ from the formula in equation \eqref{eq:ret}.
$$b_i^{\WPB}(v_i) = v_i - \int_0^{v_i} \frac{x_i(u)}{x_i(v_i)} du$$

$$(b_i^{\WPB})'(v_i) = 1 - \frac{x_i(v_i)}{x_i(v_i)} - \int_0^{v_i} \frac{-x'_i(v_i) x_i(u)}{x_i(v_i)^2} du = \int_0^{v_i} \frac{x'_i(v_i) x_i(u)}{x_i(v_i)^2} du \geq 0.$$
The inequality holds because $x_i(\cdot)$ is monotone non-decreasing for any implementable $\babic(\cdot)$.
\end{proof}

Lemma~\ref{lemma:symmetric1} says that regardless of  the value distributions, any implementable allocation function can be implemented with a WPB payment rule.

\subsection{Single-Item Ex-Post IR Variance Minimization} For single-item environments, \citet{waehrer1998auction} showed that the winner-pays-bid mechanism minimizes variance among all ex-post IR mechanisms, i.e. mechanisms for which $v_i \abic_i(\bv) - \pbic_i(\bv) \geq 0$ for all $\bv$. Their original proof assumes the efficient allocation, but it holds for any allocation with minimal modifications. We state the result below and provide a complete proof and discussion in Appendix \ref{sec: risk-minimizing ex-post IR}.

\begin{theorem}[Generalization of \cite{waehrer1998auction}]\label{thm:expost-ir} For any implementable allocation function $\babic(\bv)$, among all the ex-post IR payment rules that implement $\babic(\bv)$, the winner-pays-bid payment rule minimizes the expected risk $\E[g(R)]$ for any convex function $g$, where $R = \sum_i \pbic_i(\bv)$ is the random variable corresponding to the total revenue.
\end{theorem}

\section{Single-Item Auctions under Interim Individual Rationality}\label{sec:interim-ir}

We start by investigating whether the conclusion of \citet{waehrer1998auction} holds when we drop the ex-post IR requirement and only require interim IR, i.e., $\E[v_i \abic_i(\bv) - \pbic_i(\bv) \mid v_i] \geq 0$ for each $v_i$. This will allow us to analyze auctions like all-pay (AP). All-pay is a natural candidate for a variance minimizing auction since it minimizes the the variance of the payment of each individual agent.

\begin{lemma}\label{lemma:all-pay}
For any implementable allocation function $\babic(\bv)$, among all the payment rules that implement $\babic(\bv)$, the all-pay payment rule $\pbic^\AP_i(\bv) = z_i(v_i)$ minimizes the risk $\E[g(P_i)]$ where $g$ is any convex function and $P_i(\bv)$ is the random variable corresponding to the payment of agent $i$.
\end{lemma}

\begin{proof} Given any payment rule $\bpbic(\bv)$ implementing the given allocation function $\babic(\bv)$ and any convex function $g$, we want to show that $\E[g(\pbic_i(\bv))] \geq \E[g(\pbic_i^\AP(\bv))]$. The proof follows the same pattern as the proof of Theorem \ref{thm:expost-ir} but without conditioning on the winning bidder. We observe that:
$$\E[g(\pbic_i(\bv)) \mid v_i] \geq g(\E[\pbic_i(\bv) \mid v_i]) = g(\E[\pbic_i^\AP(\bv) \mid v_i]) = \E[g(\pbic_i^\AP(\bv)) \mid v_i]$$
where the first inequality is Jensen's inequality, the subsequent equality is due to the Revenue Equivalence Theorem \eqref{eq:ret} and the final equality is due to the fact that  $\pbic_i^\AP(\bv)$ is completely determined when conditioned on the type $v_i$ and hence Jensen's inequality holds with equality.
\end{proof}

While the all-pay payment rule minimizes the risk (and variance) of each individual bidder's payment, it does not minimize variance for the revenue. The reason is that the the second moment of the revenue is
$$\textstyle \E[(\sum_i \pbic_i(\bv))^2] = \sum_i \E[ \pbic_i(\bv)^2] + 2 \sum_{i \neq j} \E[\pbic_i(\bv) \pbic_j(\bv)].$$
The first term corresponds to the second moment of each individual bidder. For ex-post IR auctions, since only the winner pays, the second term disappears. For all-pay auctions (and non ex-post IR auction more generally) the second term is not identically zero, and, in fact, tends to dominate the variance.

The proof of Theorem \ref{thm:expost-ir} crucially relies on ex-post IR on two different points, but interestingly, the proof makes no requirement about the distribution of values except for independence. The analysis for non ex-post IR auctions is more nuanced and the answer depends on whether the setting is symmetric or asymmetric.

\subsection{Symmetric Settings: IID Valuations}\label{sec:interim-single-item-symmetric}

We start by considering an i.i.d. environment with the efficient allocation (with any tie breaking rule). In that case, we will recover the conclusion of \citet{waehrer1998auction} through a different argument based on a novel revenue decomposition technique.

\begin{lemma}\label{lemma:symmetric2}  For any i.i.d. environment, the revenue of the efficient allocation $\babic(\bv)$ coupled with the WPB payment function can be written as:
$$\textstyle\sum_i \pbic_i^{\WPB}(\bv) = \max_i b^{\WPB}(v_i)$$
where $b^{\WPB}(v_i) := z_i(v_i) / x_i(v_i)$.\footnote{We drop the subscript $i$ from $b^{\WPB}$ since the $b^{\WPB}$ is the same across all bidders in the i.i.d. environment.}
\end{lemma}

\begin{proof} Since the allocation is efficient, the winner is the bidder with the highest value; therefore:
$$\textstyle \sum_i \pbic_i^{\WPB}(\bv) = b^{\WPB}(\max_j v_j) = \max_i b^{\WPB}(v_i)$$
where the second equality follows from the monotonicity of $b^{\WPB}$ established in Lemma~\ref{lemma:symmetric1}.
\end{proof}

\begin{theorem}\label{thm:symmetric} For any i.i.d. environment, among all payment rules that implement the efficient allocation $\babic(\bv)$, the winner-pays-bid payment rule minimizes the risk $\E[g(R)]$, for any convex function $g$, where $R = \sum_i \pbic_i(\bv)$ is the random variable corresponding to the total revenue.
\end{theorem}

\begin{proof}
It is again without loss of generality to focus on BIC mechanisms. We will show that given any payment rule $\pbic_i(\bv)$ that implements the efficient allocation, we have $\E[g(\sum_j \pbic_j(\bv))] \geq \E[g(\sum_j \pbic_j^{\WPB}(\bv))]$. We start by applying convexity:
$$\textstyle g(\sum_j \pbic_j(\bv)) \geq g(\sum_j \pbic_j^{\WPB}(\bv)) + g'(\sum_j \pbic_j^{\WPB}(\bv)) \cdot (\sum_j \pbic_j(\bv) - \sum_j \pbic_j^{\WPB}(\bv))$$
Taking expectations over the expression above, we observe that it is enough to establish that:
\begin{equation}\label{eq:convexity}
\textstyle   \E[g'(\sum_j \pbic_j^{\WPB}(\bv)) \cdot (\sum_i \pbic_i(\bv) - \sum_i \pbic_i^{\WPB}(\bv))] \geq 0
\end{equation}
In order to show equation \eqref{eq:convexity}, we will start by using Lemma \ref{lemma:symmetric2} to decompose the derivative term for each bidder $i$ as follows:
\begin{equation}\label{eq:rev_decomposition}
\textstyle g'(\sum_j \pbic_j^{\WPB}(\bv)) =\lambda_i(v_i)  + \mu_i(\bv)
\end{equation}
for $$\lambda_i(v_i) = g'(b^{\WPB}(v_i)) \quad \text{and} \quad \mu_i(\bv) = g'(\max_j b^{\WPB}(v_j)) - g'(b^{\WPB}(v_i))$$
Now, we can re-write the left hand side of equation \eqref{eq:convexity} as:
$$\begin{aligned}
& \textstyle \E[g'(\sum_j \pbic_j^{\WPB}(\bv)) \cdot (\sum_i \pbic_i(v) - \sum_i \pbic_i^{\WPB}(\bv))] \\ 
& \qquad  = 
\textstyle\sum_i  \E[g'(\sum_j \pbic_j^{\WPB}(\bv)) (\pbic_i(\bv) - \pbic_i^{\WPB}(\bv)) ]  \\
& \qquad  =  \textstyle \sum_i  \E[ (\lambda_i(v_i)  + \mu_i(\bv)) (\pbic_i(\bv) - \pbic_i^{\WPB}(\bv)) ]\\
& \qquad  = \textstyle \sum_i  \E[ \lambda_i(v_i) (\pbic_i(\bv) - \pbic_i^{\WPB}(\bv)) ] + 
\sum_i  \E[\mu_i(\bv) (\pbic_i(\bv) - \pbic_i^{\WPB}(\bv)) ]
\end{aligned}$$
where the first and last equalities are re-arrangement of terms and the second equality is the revenue decomposition in equation \eqref{eq:rev_decomposition}. We conclude the proof by showing that each of the terms in the last expression is non-negative.

For the first term, we observe that $\lambda_i(v_i)$ depends only on $v_i$, so we can apply the Revenue Equivalence Theorem \eqref{eq:ret} as follows:
$$\E[ \lambda_i(v_i) (\pbic_i(\bv) - \pbic_i^{\WPB}(\bv)) ] = \E_{v_i} [ \lambda_i(v_i) \cdot \E[\pbic_i(\bv) - \pbic_i^{\WPB}(\bv) \mid v_i]] = 0$$
since $\E[\pbic_i(\bv) \mid v_i] = \E[\pbic_i^{\WPB}(\bv) \mid v_i]$.

For the second term we observe that $\mu_i(\bv) \geq 0$, since by convexity $g'$ is monotone non-decreasing. Also, $\mu_i(\bv) = 0$ whenever $i$ is the winner. Since $\pbic_i^{\WPB}(\bv) = 0$ whenever $i$ doesn't win, it follows that:
$$\mu_i(\bv) \pbic_i^{\WPB}(\bv) = 0 $$
almost surely. Hence we can write the second term as:
$$\E[\mu_i(\bv) (\pbic_i(\bv) - \pbic_i^{\WPB}(\bv)) ] = \E[\mu_i(\bv) \pbic_i(\bv)] \geq 0,$$
since we maintain the assumption of non-negative payments, thus concluding the proof. \end{proof}

\begin{corollary} Let $\babic(\bv)$ be the allocation that selects the highest bidder if it is above a fixed reserve price $r$ and otherwise selects no one. For any i.i.d. environment, the winner-pays-bid payment rule minimizes risk $\E[g(R)]$ for any convex function $g$.
\end{corollary}

\begin{proof} Same as the previous theorem except that we re-define $b^{\WPB}(v_i)$ such that it is zero for $v_i < r$.
\end{proof}

\subsection{Asymmetric Settings}\label{sec:interim-single-item-asymmetric}

In asymmetric environments, where bidders’ value distributions may differ, the winner-pays-bid rule no longer minimizes revenue variance under interim individual rationality. We first provide an example of an auction with strictly lower variance. We then derive optimality conditions that characterize variance-minimizing payment rules. Finally, we show that there always exists a payment rule satisfying the optimality conditions via a functional fixed point of a certain operator on the space of probability measures over a compact interval.

\subsubsection{Illustrative Example}

We consider an example with two independent bidders with CDF $F_1(v_1) = v_1$ and $F_2(v_2) = v_2^2$ on $[0,1]$. In the efficient allocation, the interim allocation probabilities are: $x_1(v_1) = F_2(v_1) = v_1^2$ and $x_2(v_2) = F_1(v_2) = v_2$ and the interim payment are given by the Myerson integral in equation \eqref{eq:ret}:

$$z_1(v_1) = \int_0^{v_1} (x_1(v_1) - x_1(u)) du = \frac{2}{3} v_1^3 \qquad 
  z_2(v_2) = \int_0^{v_2} (x_2(v_2) - x_2(u)) du = \frac{1}{2} v_2^2 $$

Now, the winner-pays-bid auction charges $b_1(v_1) = z_1(v_1) / x_1(v_1)$ to agent $1$ whenever $v_1 \geq v_2$ and $b_2(v_2) = z_2(v_2) / x_2(v_2)$ to agent $2$ whenever $v_2 > v_1$. Hence we must have: $b_1(v_1) = 2v_1/3$ and $b_2(v_2) = v_2/2$. Using this, we can compute the  second moment of the revenue of WPB:

$$\int_0^1 f_1(v) x_1(v) b_1(v)^2 dv + \int_0^1 f_2(v) x_2(v) b_2(v)^2 dv \approx 0.188$$

Now, we describe a payment rule that is a BIC-implementation of the optimal allocation and has strictly smaller variance. Consider the payments:

$$ \pbic_1^*(v_1, v_2) = \pi_1(v_1)\cdot \one\{  h(v_1) \geq v_2 \} \qquad 
\pbic_2^*(v_1, v_2) = \pi_2(v_2) \cdot \one\{  v_2 > h(v_1) \}$$
for functions:
$$h(u) = u\left(\frac{4}{3}\right)^{1/4}  \qquad \pi_1(v_1) = \left\{ \begin{aligned} 
& v_1 /\sqrt{3}, & & \text{for } v_1 \leq (3/4)^{1/4} \\
& 2 v_1^3/3, & & \text{for } v_1 > (3/4)^{1/4}
\end{aligned} \right.
\qquad
\pi_2(v_2) = \frac{v_2 }{2} \left(\frac{4}{3}\right)^{1/4} $$
The auction is depicted in Figure \ref{fig:asymmetric}. It has the non-standard feature that in a small region (labelled B in the figure), bidder 2 gets the item but bidder 1 pays $\pi_1(v_1)$. First, we can check that this payment rule coupled with the efficient allocation in a BIC auction by checking that the revenue equivalence holds:
$$\int_0^1 f_2(v_2) \pbic_1^*(v_1, v_2)  dv_2 = z_1(v_1) \qquad \int_0^1 f_1(v_1) \pbic_2^*(v_1, v_2) dv_1 = z_2(v_2)$$
And we can compute its second moment as:
\begin{equation*}
\int_0^{h^{-1}(1)} f_1(v_1) F_2(h(v_1)) \pi_1(v_1)^2 dv_1 +
\int_{h^{-1}(1)}^1 f_1(v_1) z_1(v_1)^2 dv_1 + \int_0^1 f_2(v_2) F_1(h^{-1}(v_2)) \pi_2(v_2)^2 dv_2,
\end{equation*}
which is approximately 0.186.

\begin{figure}[h]
\begin{center}
\begin{tikzpicture}[scale=5]
\draw (0,0) rectangle (1,1);
\draw (0,0) -- (1,1);
\node at (.5, -.1) {$v_1$};
\node at (-.1, .5) {$v_2$};
\node at (.83, 1.08) {$h(v)$};

\fill[color=red, opacity=.2] (0.000000,0.000000)--(0.93,1) -- (1,1) -- cycle;

\draw[dashed] (0.000000,0.000000)--(0.93,1);
\node at (.7,.3) {$A$};
 \node at (.75,.81) {$B$};
\node at (.3,.7) {$C$};
\end{tikzpicture}
\end{center}
\caption{In the efficient mechanism that minimizes revenue variance for $2$ asymmetric bidders with PDF $F_1(v_1) = v_1$ and $F_2(v_2) = v_2^2$, we allocate to bidder $1$ in regions A and to bidder $2$ in regions B+C, but agent $1$ pays in regions A+B while bidder $2$ pays in regions B+C. Hence, in the shaded region, bidder $2$ is allocated but bidder $1$ pays.}
\label{fig:asymmetric}
\end{figure}

\subsubsection{Risk-Minimizing Mechanism}

The intuition why WPB is not optimal in non-i.i.d. settings is that $b_1(u) > b_2(u)$ so it no longer holds that $\pbic_1(v_1, v_2) + \pbic_2(v_1, v_2) = \max(b_1(v_1), b_2(v_2))$, which is what is driving optimality in Theorem \ref{thm:symmetric}. A solution is to increase the region where bidder $1$ pays so that we can spread their payments along a larger region. With the new design, we have that $\pbic_1^*(v_1, v_2) + \pbic_2^*(v_1, v_2) = \max(\pi_1(v_1), \pi_2(v_2))$. This implies that $\bpbic^*$ is actually the variance minimizing payment rule for this environment.

More generally we can prove optimality below. Later, in Theorem \ref{theorem: existence pi}, we show that there always exist a payment rule satisfying the optimality conditions given in Theorem \ref{thm:asymmetric}.

\begin{theorem}[Optimality Condition]\label{thm:asymmetric}
For any implementable allocation function $\babic(\bv)$, given functions $\pi_i(v_i)$, let $\bpbic^*(\bv)$ be the payment rule where the agent with largest $\pi_i(v_i)$ pays $\pi_i(v_i)$ (breaking ties arbitrarily) and the remaining agents pay zero. If $\bpbic^*$ satisfies the Revenue Equivalence Theorem \eqref{eq:ret}, for any convex function $g$, $\E[g(\sum_i \pbic_i^*(\bv))] \leq \E[g(\sum_i \pbic_i(\bv))]$ for any payment rule $\bpbic(\bv)$ implementing allocation $\babic(\bv)$.
\end{theorem}

\begin{proof}  Let $\bpbic(\bv)$ be any payment rule implementing the allocation function and let $\bpbic^*(\bv)$ be a payment rule satisfying the optimality conditions in the theorem statement. To show that $\E[g(\sum_i \pbic_i^*(\bv))] \leq \E[g(\sum_i \pbic_i(\bv))]$, it is enough to show that $\E[g'(\sum_i \pbic_i^*(\bv)) (\sum_i \pbic_i(\bv) - \sum_i \pbic_i^*(\bv))] \geq 0$ by the argument in the proof of Theorem  \ref{thm:symmetric}. We decompose the revenue as follows:
$$\textstyle g'(\sum_i \pbic_i^*(\bv)) = \lambda_i(v_i) + \mu_i(\bv)$$
where:
$$\lambda_i(v_i) = g'(\pi_i(v_i)) \qquad \mu_i(\bv) = g'(\max_j \pi_j(v_j)) - g'(\pi_i(v_i))$$
As in the proof of Theorem  \ref{thm:symmetric}, we can write:
$$\textstyle \E[g'(\sum_i \pbic_i^*(\bv)) (\sum_i \pbic_i(\bv) - \sum_i \pbic_i^*(\bv))] = \textstyle \sum_i  \E[ \lambda_i(v_i) (\pbic_i(\bv) - \pbic_i^{*}(\bv)) ] + \sum_i  \E[\mu_i(\bv) (\pbic_i(\bv) - \pbic_i^{*}(\bv)) ] $$
The first term on the right-hand side is zero by the RET. Since $\mu_i(\bv) = 0$ when $\pbic^*_i(\bv) > 0$, we have $\mu_i(\bv) \pbic^*_i(\bv) = 0$. Hence the second term on the right-hand side is: $\sum_i \E[\mu_i(\bv) \pbic_i(\bv)] \geq 0$, since we maintain non-negativity of payments.
\end{proof}

\begin{theorem}[Existence]\label{theorem: existence pi}
For any implementable allocation function $\babic(\bv)$ such that the expected payment $z_i(v_i)$ is continuous and strictly increasing for every agent $i$, there exist functions $\pi_i$ satisfying the conditions of the previous theorem.
\end{theorem}

\begin{proof}
    First, observe existence of $\pi_i$ for each agent $i$ requires
    \begin{equation*}
        z_i(v_i) = \pi_i(v_i) \prod_{j \neq i} G_j(\pi_i(v_i))
    \end{equation*}
    where $G_j$ is the CDF of the random variable $\pi_j(v_j)$ induced by $F_j$. We thus seek a solution to the functional equation, where the unknowns are both $\pi_i$ and $G_j$ for all agents $i,j$.

    Fix a large enough $M > 0$ such that $z_i(v_i) \leq M$ for each agent $i$ and value $v_i$. Note that this is always possible since $z_i(v_i) \leq \overline{v}_i < \infty$ and we have a finite number of agents. Let $\mathcal{M}_1([0,M])$ be the set of probability measures on $[0,M]$; we will often write $\mathcal{M}_1$ for this set. For any $\mu \in \mathcal{M}_1$, define
    \begin{equation*}
        G^{\mu}(x) = \mu([0,x]).
    \end{equation*}
    Note that $G^{\mu}$ is non-decreasing, right-continuous, the left limit $G^{\mu}(x_{-}) = \mu([0,x))$ exists for each $x \in [0,M]$, and we have $G^{\mu}(0) = 0$ and $G^{\mu}(M) = 1$. Hence, $G^{\mu}$ satisfies the conditions of a CDF.
    
    Given a profile $\mu = (\mu_1, \dots, \mu_n) \in \mathcal{M}_1^n$, define for each agent $i$ function
    \begin{equation*}
        H_i^{\mu}(s) \equiv s \prod_{j \neq i} G^{\mu_j}(s) \qquad s \in [0,M].
    \end{equation*}
    Because each $G^{\mu_j}$ is non-decreasing, right-continuous, and non-negative, $H_i^{\mu}$ is non-decreasing and right-continuous, and $H_i^{\mu}(0) = 0$ and $H_i^{\mu}(M) = M$. We define
    \begin{equation*}
        \pi_i^{\mu}(v_i) \equiv \inf \{s \in [0,M]: H_i^{\mu}(s) \geq z_i(v_i)\}.
    \end{equation*}
    Note that this is well-defined since $H_i^{\mu}(M) = M \geq z_i(v_i)$ for any $v_i$. Moreover, $\pi_i^{\mu}$ is non-decreasing since $z_i$ and $H_i^{\mu}$ are non-decreasing. It follows $\pi_i^{\mu}$ is Borel measurable. It also follows $\pi_i^{\mu}(v_i)$ defines a random variable.

    Let $\Tilde{\mu}_i$ be the law (push-forward measure) of the random variable $\pi_i^{\mu}(v_i)$. Note that $\Tilde{\mu}_i \in \mathcal{M}_1$. Define operator $T: \mathcal{M}_1^n \to \mathcal{M}_1^n$ by
    \begin{equation*}
        T(\mu) = (\Tilde{\mu}_1, \dots, \Tilde{\mu}_n).
    \end{equation*}
    We equip $\mathcal{M}_1$ with the topology of weak convergence and $\mathcal{M}_1^n$ with the corresponding product topology. 
    
    Next, we show that $T$ is continuous. To this end, suppose $\mu^k = (\mu_1^k, \dots, \mu_n^k) \xrightarrow{w} \mu = (\mu_1, \dots, \mu_n)$. Fix an arbitrary agent $i$. Note that for any $s \in [0,M]$ that is a continuity point of $G^{\mu_j}$, for every $j \neq i$, we have
    \begin{equation*}
        H_i^{\mu^k}(s) = s \prod_{j \neq i} G^{\mu_j^k}(s) \to s \prod_{j \neq i} G^{\mu_j}(s) = H_i^{\mu}(s).
    \end{equation*}
    Since the set of discontinuities of $G^{\mu_j}$ for each agent $j$ is countable, it follows the set of discontinuities of $H_i^{\mu}$ is also countable. Thus, $H_i^{\mu^k} \to H_i^{\mu}$ point-wise on a dense set. Let $D(H_i^{{\mu}-1})$ denote the set of discontinuities of the generalized inverse of $H_i^{\mu}$. The set is countable.\footnote{Since $H_i^{\mu}$ is non-decreasing, so is $H_i^{\mu-1}$. A non-decreasing function can only have jump discontinuities. Moreover, every jump means an interval gap, and so there are countably many jumps as there can only be countably many disjoint intervals. Hence, there are countably many discontinuities.} Since we assume $z_i$ is continuous and strictly increasing, letting $Y_i = z(v_i)$, we must have $\mathbb{P}(Y_i \in D(H_i^{\mu-1})) = 0$. Hence, by Lemma \ref{lemma: convergence of generalized inverse} in the Appendix, we have
    \begin{equation*}
        \pi_i^{\mu^k}(v_i) = H_i^{\mu^k-1}(Y_i) \to H_i^{\mu-1}(Y_i) = \pi_i^{\mu}(v_i) \qquad a.s.
    \end{equation*}
    Moreover, since $0 \leq \pi_i^{\mu^k} \leq M$, we must have that the law of $\pi_i^{\mu^k}(v_i)$ converges weakly to $\pi_i^{\mu}(v_i)$. It follows that $T_i: \mathcal{M}_1^n \to \mathcal{M}_1$ given by $\mu \mapsto \Tilde{\mu}_i$ is continuous. Repeating this for all agents $i$ yields continuity of $T$.

    We show in Lemma \ref{lemma: fixed point} in the Appendix that $T$ admits a fixed point $\mu^{*}$ by applying the Schauder fixed point theorem (Lemma \ref{lemma: schauder}). In particular, we embed $\mathcal{M}_1$ endowed with the bounded Lipschitz metric, which metrizes weak convergence, in the dual of the space of bounded Lipschitz functions on $[0,M]$ endowed with the bounded Lipschitz norm, which is a Banach space. We show we can do so via an injective isometry and prove the image of $\mathcal{M}_1$ under this mapping is compact in the ambient space. Finally, we apply the Schauder fixed point theorem to a transformation of $T$ under the isometry, and recover a fixed point $\mu^{*}$ of $T$ by taking its inverse. The details are contained in Appendix \ref{appendix: existence of pi}.
    
    To conclude, we define $\pi_i \equiv \pi_i^{\mu^{*}}$, for each agent $i$. By construction, the function has the desired properties for each agent $i$.
\end{proof}

\section{Multi-Item Auctions}

We now turn to the allocation of $k \ge 2$ items. Under interim individual rationality, we show that in i.i.d.\ regular environments there exists a mechanism that implements the efficient allocation and achieves zero revenue variance. Under ex-post IR, we provide an example demonstrating that revenue variance can be reduced relative to winner-pays-bid formats by introducing negative correlation in the payments of winning bidders. Finally, we show that the revenue variance ranking between the winner-pays-bid format and the uniform $(k+1)$-st price auction extends from single-item to multi-item environments.

\subsection{Interim Individually Rational Mechanisms}

We begin with interim participation constraints. With a single item, the variance minimizing payment rule in Section \ref{sec:interim-ir} had the property that at most one bidder pays, causing the cross terms $\E[P_i P_j]$ in $\E[(\sum_i P_i)^2]$ to cancel. With more than one item, it is possible to take this idea further and create negative correlations between payments. We first exemplify how this can be done for the uniform distribution and then generalize to any regular distribution. 


\subsubsection{Example for the Uniform Distribution}

We start with an example of $2$-items and $3$-bidders with values distributed i.i.d. from $U[0,1]$. In that case, we exhibit a payment rule satisfying interim IR, RET and non-negative payments with zero variance. We define a payment function that depends only on the ordered valuations $v_{(1)} \ge v_{(2)} \ge v_{(3)}$:

\[ P_{(1)}(\bv) = \frac{2-v_{(2)}}{4}, \quad P_{(2)}(\bv) = \frac{v_{(2)}}{4}, \quad P_{(3)}(\bv) = 0. \]

Observe that since $P_{(1)}(\bv) + P_{(2)}(\bv)$ is a constant, the revenue has variance zero so this payment rule is clearly optimal if it is feasible. Non-negativity of the payment follows from $0 \leq v_{(2)} \leq 1$. We are left to check that it satisfies RET:
$$\begin{aligned}  \int_{0}^{v} \frac{2-u}{4} (2u)  du + \int_{v}^{1} \frac{v}{4} (2v)  du & = \frac{1}{2} \int_0^v (2u-u^2) du + \frac{v^2}{2} \int_v^1 du  \\ & = \left(\frac{v^2}{2} - \frac{v^3}{6}\right) + \left(\frac{v^2}{2} - \frac{v^3}{2}\right) = v^2 - \frac{2}{3}v^3 = z(v). \end{aligned}$$

At first glance, this payment rule resembles that of \citet{eso1999auction}, who study variance minimization subject to interim individual rationality but relax the requirement that payments be non-negative; see Appendix~\ref{sec:eso_futo} for a discussion. Their mechanism, however, relies on allowing negative transfers and on all bidders making non-zero payments. By contrast, our mechanism features non-negative payments and requires transfers only from winning bidders.

\subsubsection{Optimal Payment Rule for Regular Distributions}

We generalize the aforementioned construction to any (Myerson-)regular distribution. We show that it is possible to correlate the payments such that the variance is zero, while keeping the properties of interim IR, non-negative payments, and that losers pay zero. 

We recall that a distribution $F_i$ with density $f_i$ is (Myerson-)regular if its corresponding \emph{virtual value function}
$$\phi_i(v_i) = v_i - \frac{1-F_i(v_i)}{f_i(v_i)}$$ 
is strictly increasing. Notable examples are all distributions with log-concave densities like uniform, exponential and normal. We refer to Section \ref{sec: useful facts for regular distributions} in the Appendix for useful facts on regular distributions, including a proof of the following bound that will be useful in our analysis:

\begin{lemma}\label{lemma: rev and endpoint of virtual value}
    Fix an i.i.d. regular environment and with $2 \leq k < n$ and let $\bar R$ be the expected revenue of any mechanism that implements the efficient allocation and satisfies interim IR. Then:
        \begin{equation*}
        \bar{R} < -(n-k)\phi(0).
    \end{equation*}
\end{lemma}

\begin{theorem} \label{thm:general_n_k}
For any i.i.d. regular environment with continuously differentiable density $f$ and with $2 \leq k < n$, there is a non-negative payment rule that implements the efficient allocation and satisfies interim IR such that the variance of the revenue is zero.
\end{theorem}

\begin{proof}
Let $\bar R =\E[\Rev]$ be the expected total revenue of any BIC mechanism implementing the efficient allocation. We define a payment rule that distributes the target revenue $\bar{R}$ among the top $k$ bidders as a function of the $k$-th order statistic, $v_{(k)}$:
$$P_{(i)}(\bv) = \frac{\bar R - \psi(v_{(k)})}{k-1} \text{ for } i=1, \dots, k-1 \quad \text{ and } \quad P_{(k)}(\bv) =\psi(v_{(k)})$$
and the remaining losers pay zero. 

By construction, the total ex-post revenue is $\sum_{j=1}^n P_{(j)}(\bv) = (k-1) \frac{\bar R - \psi(v_{(k)})}{k-1} + \psi(v_{(k)}) = \bar R$ for every realization of the random variables, so the revenue variance is identically zero. We must now show that $\psi(v)$ can be chosen to satisfy RET, and that $0 \leq \psi(v) \leq \bar R$ so that all payments are non-negative.

A bidder with value $v$ is the $k$-th highest bidder with probability $Q(v) = \binom{n-1}{k-1} (1-F(v))^{k-1} F(v)^{n-k}$. Let $\tilde{G}(u) = \binom{n-1}{k-2, 1, n-k} (1-F(u))^{k-2} F(u)^{n-k} f(u)$ be the probability density that some other bidder is the $k$-th highest at value $u$ and write: $\tilde{G}(u) = (k-1) G(u)$ where $G(u) = \binom{n-1}{k-1} (1-F(u))^{k-2} F(u)^{n-k} f(u)$.  The RET condition defining the interim expected payment $z(v)$ for the efficient allocation can be written as:
$$ z(v) = \int_0^v \left(\frac{\bar{R} - \psi(u)}{k-1}\right) (k-1) G(u) du + \psi(v) Q(v).$$
Differentiating with respect to $v$ and rearranging terms, we obtain a first-order linear ordinary differential equation for $\psi(v)$:
\begin{equation} \label{eq:ode_general}
\psi'(v) Q(v) + \psi(v) [Q'(v) - G(v)] = z'(v) - \bar{R} G(v).
\end{equation}
By algebraic manipulation, $Q'(v) - G(v) = \binom{n-1}{k-1} F(v)^{n-k-1} (1-F(v))^{k-2} f(v) [n - k - n F(v)]$. The ODE admits an exact solution via the integrating factor method:
\begin{equation}\label{eq: solution ODE}
    \psi(v) = \frac{\int_0^v [z'(u) - \bar{R} G(u)] (1 - F(u)) du}{\binom{n-1}{k-1} F(v)^{n-k} (1 - F(v))^k} = \frac{N(v)}{D(v)}    .
\end{equation}
To complete the proof, we must show that $0 \leq \psi(v) \leq \bar R$ for every $v$, so that payments are non-negative.

\medskip

\noindent \emph{Step 1 ($0 \leq \psi$):}  We show this in three steps: (i) $\psi(0) = 0$, (ii) the numerator $N(v)$ is unimodal, and (iii) $N(\overline{v}) = 0$.

We first show $\psi(0) = 0$. Since $N(0) = D(0) = 0$, we evaluate the limit using L'Hôpital's rule: $\lim_{v \downarrow 0} \psi(v) = \lim_{v \downarrow 0} \frac{N'(v)}{D'(v)}$. We show this limit is zero in Lemma \ref{lemma: psi 0} in Appendix \ref{sec: appendix proof multi-item interim IR}.

Next, we show the numerator is unimodal. We can write the derivative of the numerator as follows (see the proof of Lemma \ref{lemma: psi 0} in Appendix \ref{sec: appendix proof multi-item interim IR}):
$$N'(u) = \binom{n-1}{k-1} F(u)^{n-k-1} (1-F(u))^{k-1} f(u) \int_0^u f(x) [-(n-k)\phi(x)  - \bar R ] dx,$$
where the integral substitution uses the identity $\frac{d}{dx}[x(1-F(x))] = 1-F(x)-xf(x) = -f(x)\phi(x)$.  Since the virtual value is monotonically increasing, the bracketed term in the integrand decreases from a strictly positive quantity to a strictly negative quantity by Lemma \ref{lemma: rev and endpoint of virtual value}. That is, $N$ is unimodal.

Finally, since $N$ is unimodal and $N(0)=0$, it suffices to show that $N(\overline{v}) = 0$. First, integration by parts reveals that $\int_0^{\overline{v}} z'(u) (1-F(u)) du = \int_0^{\overline{v}} z(u) f(u) du = \E[z(v)]$. Since the mechanism is efficient, the ex-ante expected payment of a single bidder is $\bar{R}/n$. 
For the second integral, using the substitution $y=F(u)$ and the Beta function identity $\int_0^1 y^{n-k} (1-y)^{k-1} dy = \frac{(n-k)!(k-1)!}{n!}$:
$$ \int_0^{\overline{v}} \bar{R}G(u) (1-F(u)) du = \bar{R} \binom{n-1}{k-1} \int_0^1 y^{n-k} (1-y)^{k-1} dy = \frac{\bar{R}}{n}.$$
Thus, $N(\overline{v}) = \frac{\bar{R}}{n} - \frac{\bar{R}}{n} = 0$. It follows $N(v) \ge 0$ for all $v$, concluding the proof of Step 1.

\medskip

\noindent \emph{Step 2 ($\psi \le \bar R$):} We have $\psi(0) = 0$ and $\lim_{v \to \overline{v}} \psi(v) = \frac{\bar R}{k} < \bar R$ by L'Hopital's rule:
$$ \lim_{v \to \overline{v}} \psi(v) = \lim_{v \to \overline{v}} \frac{N'(v)}{D'(v)} = \lim_{v \to \overline{v}} \frac{v(n-k)(1-F(v)) - \bar R F(v)}{n-k - n F(v)} = \frac{0 - \bar R}{-k} = \frac{\bar R}{k}. $$
Since $\psi$ is differentiable everywhere by the maintained hypotheses, any violation of the upper bound $\bar R$ must occur at an interior critical point $\psi'(v^{*}) = 0$. We show there is no such violation by considering the following two cases.\\

\noindent \emph{Case 1: $F(v^*) \geq \frac{n-k}{n}$.} 
In this region, $Q'(v^*) - G(v^*) \leq 0$. Substituting the critical point condition $\psi'(v^*) = 0$ into ODE \eqref{eq:ode_general}, the upper-bound condition $\psi(v^*) \leq \bar R$ simplifies to $z'(v^*) \geq \bar{R} Q'(v^*)$. 

For the sub-region $F(v^*) \geq \frac{n-k}{n-1}$, we have $Q'(v^*) \leq 0$. Since $\bar{R} > 0$ and $z'(v^*) \ge 0$, the inequality $\bar R Q'(v^*) \leq 0 \leq z'(v^*)$ holds trivially. 

For the sub-region $F(v^*) \in [\frac{n-k}{n}, \frac{n-k}{n-1})$, $Q'(v^*) > 0$, so we can divide by $Q'(v^*)$ to rewrite the necessary condition as:
\begin{equation} \label{eq:R_bound}
\bar R \leq \frac{z'(v^*)}{Q'(v^*)} = \frac{v^* (n-k)(1-F(v^*))}{n-k - (n-1)F(v^*)} \equiv Y(v^*).
\end{equation}

To show that this condition holds for all $v^*$ in the specified range, we first prove that the expression on the right-hand side, $Y(v^*)$, is strictly monotonically increasing. Differentiating $Y(v)$  yields:
$$\begin{aligned}
Y'(v) &= \frac{(n-k)[1-F(v) - v f(v)](n-k-(n-1)F(v)) - v(n-k)(1-F(v))[-(n-1)f(v)]}{[n-k - (n-1)F(v)]^2} \\
&= \frac{(n-k)(1-F(v))(n-k-(n-1)F(v)) + v(n-k) f(v) (k-1)}{[n-k - (n-1)F(v)]^2}.
\end{aligned}$$
Since $k \ge 2$ and $F(v) < \frac{n-k}{n-1}$ in this sub-region, every algebraic term in the numerator is strictly positive so $Y'(v) > 0$.
Because $Y(v^*)$ is strictly increasing in this interval, its minimum occurs at the lower boundary, $v_0 = F^{-1}(\frac{n-k}{n})$. Substituting $F(v_0) = \frac{n-k}{n}$ into $Y(v_0)$ simplifies the fraction exactly: 
$$ Y(v_0) = \frac{v_0 (n-k) \left(1-\frac{n-k}{n}\right)}{n-k - (n-1)\left(\frac{n-k}{n}\right)} = \frac{v_0 (n-k) \left(\frac{k}{n}\right)}{\frac{n-k}{n}} = k v_0.$$
Thus, to satisfy equation \eqref{eq:R_bound} for all $v^*$ in the interval, it is sufficient to prove the revenue bound $\bar R \le k v_0$.

To prove this, let $x(v)$ be the interim probability of allocation. By Myerson's lemma, we can write: $\bar R = n \int_0^{\overline{v}} \phi(v) x(v) f(v) dv$. Also by Myerson's lemma:$\int_{v_0}^{\overline{v}} \phi(v) f(v) dv = v_0 (1-F(v_0))$ (see Appendix \ref{sec: useful facts for regular distributions}) which allows us to write: 
$$ \bar R - k v_0 = n \int_0^{\overline{v}} \phi(v) f(v) [ x(v) - \mathbf{1}_{v > v_0} ] dv.$$

Observe that since $0 \leq x(v) \leq 1$, the bracketed term $x(v) - \mathbf{1}_{v > v_0}$ has a single-crossing property at $v_0$: it is non-negative for $v < v_0$ and non-positive for $v > v_0$. Moreover, the virtual value $\phi(u)$ is monotonically increasing. Replacing $\phi(u)$ with the constant $\phi(v_0)$ strictly increases the value of the integral (by shifting the negative mass upwards and the positive mass downwards). Hence:
$$ \bar R - k v_0 \leq \phi(v_0) n \int_0^{\overline{v}} f(v) [ x(v) - \mathbf{1}_{v > v_0} ] dv = 0,$$
since $\int_0^{\overline{v}} f(v) x(v) dv = k/n$ since it is the ex-ante probability of allocation of any given bidder and $\int_{v_0}^{\overline{v}} f(v) dv = 1-F(v_0)=k/n$.

This completes the proof that $\bar R \le k v_0$, validating the upper bound constraint $\psi(v^*) \le \bar R$ in this region.

\medskip

\noindent \emph{Case 2: $F(v^*) < \frac{n-k}{n}$.} 
We rewrite the ODE in \eqref{eq:ode_general} by dividing by the strictly positive function $Q(v)$: 
$$ \psi'(v) = \frac{Q'(v)-G(v)}{Q(v)} [H(v) - \psi(v)] \qquad \text{where} \qquad H(v) = \frac{z'(v) - \bar{R}G(v)}{Q'(v)-G(v)}.$$
The function $H(v)$ acts as a moving target that $\psi(v)$ chases. At any critical point $v^*$ where $\psi'(v^*) = 0$, the function intersects the target, so $\psi(v^*) = H(v^*)$. To determine the nature of these critical points, we differentiate the ODE with respect to $v$ to find the second derivative at the critical point. By the product rule:
$$ \psi''(v^*) = \left[\frac{d}{dv} \left( \frac{Q'-G}{Q} \right)\right]_{v^*} [H(v^*) - \psi(v^*)] + \frac{Q'(v^*)-G(v^*)}{Q(v^*)} [H'(v^*) - \psi'(v^*)].$$
Since $H(v^*) - \psi(v^*) = 0$ and $\psi'(v^*) = 0$ at the critical point, the first term vanishes and the expression simplifies exactly to:
$$ \psi''(v^*) = \frac{Q'(v^*)-G(v^*)}{Q(v^*)} H'(v^*).$$
In this sub-region, $F(v^*) < \frac{n-k}{n}$, meaning the numerator of the pre-factor $Q'(v^*) - G(v^*) > 0$. Since the probability $Q(v^*) > 0$ for all interior points, the entire pre-factor is strictly positive. Thus, the sign of the second derivative $\psi''(v^*)$ is determined entirely by the sign of $H'(v^*)$.

To evaluate $H'(v)$, we first simplify $H(v)$ by canceling the common terms $\binom{n-1}{k-1} F^{n-k-1}(1-F)^{k-2}f$ from the numerator and denominator of the target function:
$$ H(v) = \frac{v(n-k)(1-F(v)) - \bar{R}F(v)}{n-k - n F(v)}.$$
Applying the quotient rule and simplifying the resulting cross-terms in the numerator, we obtain:
$$ H'(v) = \frac{(n-k)}{(n-k-nF(v))^2} \left[ (1-F(v))(n-k-nF(v)) + f(v)(k v - \bar R) \right].$$
We express the density $f(v)$ in terms of the Myerson virtual value via the identity $f(v) = \frac{1-F(v)}{v-\phi(v)}$. Substituting this into the brackets and factoring out the strictly positive term $\frac{1-F(v)}{v-\phi(v)}$ yields:
$$ H'(v) = \frac{(n-k)(1-F(v))}{(n-k-nF(v))^2 (v-\phi(v))} K(v) \text{ for } K(v) =  (v-\phi(v))(n-k-nF(v)) + k v - \bar R .$$
The sign of $H'(v)$ is determined solely by the sign of  $K(v)$, since all the fractional pre-factors are strictly positive for $v>0$. We establish that $K(v) > 0$ for all $v$ in this region by examining its derivative:\footnote{Note that $\phi$ is differentiable by the maintained assumptions on $f$.}
$$ K'(v) = (1-\phi'(v))(n-k-nF(v)) - n f(v)(v-\phi(v)) + k.$$
Using the identity $f(v)(v-\phi(v)) = 1-F(v)$, this simplifies to:
$$ K'(v) = (n-k-nF(v)) - \phi'(v)(n-k-nF(v)) - n(1-F(v)) + k = - \phi'(v)(n-k-nF(v)).$$
Since we are in the region $F(v) < \frac{n-k}{n}$, the multiplier $(n-k-nF(v))$ is strictly positive. By the regularity of $F$, the virtual value is monotonically increasing, so $\phi'(v) > 0 $. Therefore, $K'(v) < 0$.

Because $K(v)$ is monotonically decreasing, its minimum on the interval occurs at the right boundary, $v_0 = F^{-1}(\frac{n-k}{n})$. At this boundary, the first term of $K(v)$ vanishes:
$$ K(v_0) = (v_0-\phi(v_0))(0) + k v_0 - \bar R = k v_0 - \bar R.$$
By the strict revenue inequality $\bar R \le k v_0$ proven in Case 1, we have $K(v_0) > 0$. Since $K(v)$ is decreasing toward a non-negative minimum, $K(v) > 0$ everywhere in the interior of this region.

Consequently, $H'(v^*) > 0$, which implies the second derivative $\psi''(v^*) > 0$. This confirms that any critical point in this region must be a strict local minimum, not a maximum. Since there are no local maxima, the function cannot exceed the upper bound $\bar R$.
\end{proof}

\subsection{Ex-Post Individually Rational Mechanisms}

Under interim IR, we showed we can go as far as zero revenue variance while guaranteeing the mechanism never pays an agent. The mechanism, however, relied on the satisfaction of interim IR only; indeed, it can be readily seen that our example exhibits regions where the highest-value agent may be required to pay more than her value.

We now turn to settings with ex-post IR and investigate whether the conclusion of \citet{waehrer1998auction} holds when we sell more than one item. Recall that the proof technique of \citet{waehrer1998auction} strongly relied on at most one agent having a non-zero payment, which had convenient mathematical consequence of eliminating the cross terms in the second moment of the revenue  $\E[\Rev^2] = \E[\sum_i P_i^2 + \sum_{i \neq j} P_i P_j]$.

In Section \ref{sec:wpb_not_optimal} in the appendix, we show that even in the simplest environment of three bidders with i.i.d. uniform values and two items, the winner-pays-bid rule is no longer variance-minimizing. The key to reducing revenue variance is to raise the payment of the highest bidder and lower that of the second-highest bidder, thereby reducing covariance terms in total revenue.

While we are unable to characterize the optimal mechanism analytically even in this simple setting with three bidders and two items, we characterize its key properties and show how it can be computed numerically as the solution to a quadratic program. We refer the reader to Section \ref{sec:ex_post_multi_charact}  for details.

\subsection{Winner-Pays-Bid vs $(k+1)$-st Price Auction}

While the winner-pays-bid format is not variance-minimizing, the question remains how it compares to other standard formats in the multi-item case. In this section, we show that it attains lower variance of revenue than the $(k+1)$-st price auction. In what follows, let $\babic(\bv)$ be the efficient allocation rule.

To compress notation, we define the \emph{uniform price} for a value profile $\bv$ to be the $(k+1)$-st highest value: $p^u(\bv) \equiv v_{(k + 1)}$. In the \textit{uniform $(k{+}1)$-price auction}, each winner pays the uniform price $p^u(\bv)$ and losers pay $0$: 
\begin{equation*}
    \pbic_i^{U}(\bv) \equiv \abic_i(\bv)\,p^u(\bv) \qquad R^{U}(\bv) \equiv \sum_{i=1}^n \pbic_i^{U}(\bv)= k\,p^u(\bv).
\end{equation*}

We focus on the symmetric Bayesian Nash equilibrium in strictly increasing bidding strategies in winner-pays-bid formats. By the Revenue Equivalence Theorem, the equilibrium bidding function has the form $b(v_i) = \E[p^u(\bv) | v_i = v_i, \abic_i(\bv) = 1]$, for any $v_i > 0$; by individual rationality, $b(0) = 0$. Payments and revenue in equilibrium are 
\begin{equation*}
    \pbic_i^{D}(v) \equiv \abic_i(\bv)\,b(v_i) \qquad R^{D}(v)\equiv \sum_{i=1}^n \pbic_i^{D}(v).
\end{equation*}

\begin{theorem}\label{th:varcomp}
    Consider an i.i.d. environment with $1 \leq k < n$. In the symmetric strictly increasing equilibrium of the discriminatory winner-pays-bid auction,
    \begin{equation*}
        \Var[R^D(\bv)] \leq \Var[R^U(\bv)].
    \end{equation*}
\end{theorem}

\begin{proof}[Proof of Theorem \ref{th:varcomp}]
    By efficiency of the allocation, exactly $k$ agents obtain the good. By the Cauchy-Schwarz inequality,
    \begin{equation*}
        R^D(\bv)^2 = \left( \sum_{i=1}^n \abic_i(\bv) b(v_i) \right)^2 \leq k \sum_{i=1}^n \abic_i(\bv) b(v_i)^2 \qquad a.s.
    \end{equation*}
    Fix an arbitrary agent $i$. By conditional Jensen's inequality,
    \begin{equation*}
        b(v_i)^2 = \E[p^u(\bv) | v_i = v_i, \abic_i(\bv) = 1]^2 \leq \E[p^u(\bv)^2 | v_i = v_i, \abic_i(\bv) = 1] \qquad a.s.
    \end{equation*}
    Multiplying both sides by $A_i(\bv)$ and taking expectations yields
    \begin{equation*}
        \E[\abic_i(\bv) \, b(v_i)^2] \leq \E[ \abic_i(\bv) \, \E[p^u(\bv)^2 | v_i = v_i, \abic_i(\bv) = 1] ].
    \end{equation*}
    By the law of iterated expectations,
    \begin{equation*}
        E[\abic_i(\bv) \, p^u(\bv)^2] = \E[ \, \E[\abic_i(\bv) \, p^u(\bv)^2 | v_i, \abic_i(\bv)]].
    \end{equation*}
    Note that
    \begin{equation*}
        \E[\abic_i(\bv) \, p^u(\bv)^2 | v_i, \abic_i(\bv) = 1] = \E[p^u(\bv)^2 | v_i, \abic_i(\bv) = 1] \text{ and } \E[\abic_i(\bv) \, p^u(\bv)^2 | v_i, \abic_i(\bv) = 0] = 0.
    \end{equation*}
    Hence,
    \begin{equation*}
        \E[\abic_i(\bv) \, p^u(\bv)^2 | v_i, \abic_i(\bv)] = \abic_i(\bv) \, \E[ p^u(\bv)^2 | v_i, \abic_i(\bv) = 1]
    \end{equation*}
    and so
    \begin{equation*}
        \E[ \abic_i(\bv) \, \E[p^u(\bv)^2 | v_i, \abic_i(\bv) = 1] ] = E[\abic_i(\bv) \, p^u(\bv)^2].
    \end{equation*}

    Now it follows that
    \begin{equation*}
        \E[R^D(\bv)^2] \leq k \sum_{i=1}^n \E[\abic_i(\bv) \, b(v_i)^2] \leq k \sum_{i=1}^n \E[\abic_i(\bv) \, p^u(\bv)^2].
    \end{equation*}
    Since $\sum_{i=1}^n \abic_i(\bv) = k$ a.s., we have $\sum_{i=1}^n \abic_i(\bv) p^u(\bv) = k p^u(\bv)$ a.s., from which it follows that
    \begin{equation*}
        k \sum_{i=1}^n \E[\abic_i(\bv) p^u(\bv)^2] = k \E[k p^u(\bv)^2] = \E[(k p^u(\bv))^2] = \E[R^U(\bv)^2].
    \end{equation*}
    Therefore, $\E[R^D(\bv)^2] \leq \E[R^U(\bv)^2]$. Since $\E[R^U(\bv)] = \E[R^D(\bv)]$, the result follows.
\end{proof}

A corresponding result for i.i.d.\ environments with log-concave valuation distributions was previously obtained by \citet{pekevc2025variance}. Our result dispenses with the log-concavity assumption and therefore applies to a broader class of environments.

\section{Conclusion}\label{sec: conclusion}

This paper studies which ex-post payment rules, among those that implement a given allocation and satisfy various participation constraints, minimize revenue variance. Moving beyond the canonical single-item setting with ex-post individual rationality, where the winner-pays-bid format minimizes any convex risk measure of revenue \citep{waehrer1998auction}, we analyze environments with interim participation constraints and multiple items. We show that the revenue-variance optimality of winner-pays-bid formats breaks down in these richer settings and provide mechanisms that are variance-minimizing instead.

Several questions remain for future research. In multi-item environments, it would be natural to study heterogeneous-item settings, such as position auctions. More broadly, in both homogeneous- and heterogeneous-item environments, it would be valuable to understand whether the ranking of winner-pays-bid formats and the Vickrey-Clarke-Groves mechanism extends beyond revenue variance to other convex risk measures.

\bibliographystyle{ACM-Reference-Format}
\bibliography{references}

\appendix

\section{Useful Facts for Regular Distributions}\label{sec: useful facts for regular distributions}

For regular distributions, Myerson's lemma (\cite{Myerson1981}) states that the expected revenue is the same as the expected virtual value: $\E[z_i(v_i)] = \E[x_i(v_i) \phi_i(v_i)]$. In particular, if we consider the auction that allocates only when $v_i \geq r$ for a reserve price $r$ it holds that:
\begin{equation*}
    r (1-F(r)) = \int_r^{\overline{v}} \phi(v) f(v) dv.
\end{equation*}
and in particular for $r=0$ we have:
\begin{equation*}
    \int_0^{\overline{v}} \phi(v) f(v) dv = 0.
\end{equation*}

Finally, we prove Lemma \ref{lemma: rev and endpoint of virtual value}.

\begin{proof}[Proof of Lemma \ref{lemma: rev and endpoint of virtual value}]
    By RET, the expected revenue $\bar{R}$ is equal to the expected virtual surplus of the winners:
    \begin{equation*}
        \bar{R} = \mathbb{E}\left[\sum_{i=1}^k \phi(v_{(i)})\right].
    \end{equation*}
    By linearity of expectation and the fact that ex-ante expectation of virtual values is zero (see Appendix \ref{sec: useful facts for regular distributions}), the sum of the expected virtual values of all $n$ order statistics must equal $n$ times the expected virtual value of a single bidder: $\mathbb{E}\left[\sum_{i=1}^n \phi(v_{(i)})\right] = n \cdot \mathbb{E}[\phi(v)] = 0$. We can decompose this total sum into the $k$ winners and the $(n-k)$ losers:
    \begin{equation*}
        \underbrace{\mathbb{E}\left[\sum_{i=1}^k \phi(v_{(i)})\right]}_{\bar{R}} + \mathbb{E}\left[\sum_{i=k+1}^n \phi(v_{(i)})\right] = 0
    \end{equation*}
    Rearranging this identity expresses the revenue entirely in terms of the losing bidders:
    \begin{equation}\label{eq:rev_losers}
        \bar{R} = \mathbb{E}\left[\sum_{i=k+1}^n -\phi(v_{(i)})\right] 
    \end{equation}
    By regularity, $-\phi(v)$ is monotonically decreasing. Since the support is bounded below by $0$, the order statistics are non-negative ($v_{(i)} \geq 0$). By monotonicity, and since losing order statistics are strictly positive almost surely ($v_{(i)} > 0$ for $i > k$), the maximum possible value of $-\phi(v)$ occurs at the lower bound: $-\phi(v_{(i)}) < -\phi(0)$ for all $i \in \{k+1, \dots, n\}$ almost surely. Substituting this upper bound into equation (\ref{eq:rev_losers}) yields:
    \begin{equation*}
        \bar{R} < \mathbb{E}\left[\sum_{i=k+1}^n -\phi(0)\right] = -(n-k)\phi(0).
    \end{equation*}
\end{proof}

\section{Risk-Minimizing Ex-Post IR Mechanisms}\label{sec: risk-minimizing ex-post IR}

We begin with the proof of Theorem \ref{thm:expost-ir}. Subsequently, we provide direct corollaries of the result.

\begin{proof}[Proof of Theorem \ref{thm:expost-ir}.]
Given the revelation principle, it is without loss of generality to focus on BIC mechanisms. Let $\bpbic(\bv)$ be any payment rule implements $\babic(\bv)$ and let $\bpbic^{\WPB}(\bv)$ be the winner-pays-bid payment rule associated with $\babic(\bv)$. We know by \eqref{eq:ret} that $\E[\sum_j \pbic_j(\bv)] = \E[\sum_j \pbic^{\WPB}_j(\bv)]$. We want to show that $$\textstyle \E[g(\sum_j \pbic_j(\bv))] \geq \E[g(\sum_j \pbic^{\WPB}_j(\bv))]$$

We will focus on one bidder at a time and show the inequality above conditioned on $i$ having value $v_i$ and being the winner ($\abic_i(\bv) = 1$). By ex-post IR, we know that in that case: $\sum_j \pbic_j(\bv) = \pbic_i(\bv)$. Hence we have:

$$\begin{aligned}
\textstyle \E[g(\sum_j \pbic_j(\bv)) \mid v_i, \abic_i(\bv) =1 ] & = \E[g(\pbic_i(\bv)) \mid v_i, \abic_i(\bv) =1 ]
\\ & \geq g(\E[\pbic_i(v) \mid v_i, \abic_i(\bv) =1 ])
\end{aligned}$$
Where the first step is because of restriction to ex-post IR payment rules, and the second step is by Jensen's inequality. Now we observe that:
$$\E[\pbic_i(\bv) \mid v_i, \abic_i(\bv) =1 ] = \frac{\E[\pbic_i(\bv) \abic_i(\bv) \mid v_i ]}{\P[\abic_i(\bv) =1 \mid v_i]} = \frac{z_i(v_i)}{x_i(v_i)} = \E[\pbic_i^\WPB(v_i) \mid v_i, \abic_i(\bv) = 1]$$
where the second equality again relies on ex-post IR, since $\pbic_i(\bv) \abic_i(\bv) = \pbic_i(\bv)$. Plugging that in the previous equation we have:
$$\begin{aligned}
\textstyle \E[g(\sum_j \pbic_j(\bv)) \mid v_i, \abic_i(\bv) =1 ] 
& \geq g(\E[\pbic_i^\WPB(\bv) \mid v_i, \abic_i(\bv) =1 ])
\\ & = \E[g(\pbic_i^\WPB(\bv)) \mid v_i, \abic_i(\bv) =1 ])
\\ & = \textstyle \E[g(\sum_j \pbic_j^\WPB(\bv)) \mid v_i, \abic_i(\bv) =1 ])
\end{aligned}$$
The equality in the second line is the case where Jensen's inequality holds with equality since $\pbic_i^\WPB(\bv)$ is completely determined when conditioned on $v_i, \abic_i(\bv)=1$. Finally the last equality is due to the fact that only the winner pays in WPB.

Taking expectations over $v_i, \abic_i(\bv) = 1$, summing over all bidders and applying the law of total expectation concludes the proof.
\end{proof}

\begin{corollary}\label{cor:variance_minimization} In the setting of Theorem~\ref{thm:expost-ir}, WPB minimizes the variance $\Var(R) = \E[(R-\E[R])^2]$ among ex-post IR mechanisms.
\end{corollary}

\begin{proof}
We can decompose $\Var(R) = \E[R^2] - \E[R]^2$. By the Revenue Equivalence Theorem \eqref{eq:ret}, the second term is the same across all payment rules. By the previous theorem with $g(x) = x^2$, the first term is minimized for WPB.
\end{proof}

\begin{corollary}\label{cor:other_metrics} In the setting of Theorem~\ref{thm:expost-ir}, WPB minimizes  the risk of each individual bidder's payment $\E[g(P_i)]$, the risk of each individual bidder's $\E[g(U_i)]$ and the risk of the sum of all the bidders' utilities $\E[g(U)]$ among all ex-post IR payment rules.
\end{corollary}

\begin{proof}
The proof that WPB minimizes $\E[g(P_i)]$ is direct from the proof of Theorem \ref{thm:expost-ir}. The argument for utility and sum of utilities follows exactly the same proof format by observing that: (i) only the winner has non-zero utility; (ii) in WPB, the utility of an agent is completely determined when conditioned on $v_i$ and $\abic_i(\bv)=1$. With those observations the exact same argument carries through.
\end{proof}

\section{Proof of Theorem \ref{theorem: existence pi}}\label{appendix: existence of pi}

Let $M > 0$ and let $h: [0,M] \to [0,M]$ be a non-decreasing right-continuous function. Its generalized inverse is defined by
\begin{equation*}
    h^{-1}(y) \equiv \inf\{ s \in [0,M]: h(s) \geq y \}.
\end{equation*}

\begin{lemma}[Convergence of generalized inverses]\label{lemma: convergence of generalized inverse}
    Let $M > 0$ and $h_k, h: [0, M] \to [0,M]$ be non-decreasing right-continuous functions, for any $k \in \mathbb{N}$. Suppose $\lim_{k \to \infty} h_k(s) = h(s)$ for a dense set containing the set of continuity points of $h$. Also suppose $y \in [0,M]$ is a continuity point of $h^{-1}$ and $h^{-1}(y) > 0$. Then $\lim_{k \to \infty} h_k^{-1}(y) = h^{-1}(y)$.
\end{lemma}

\begin{proof}
    \noindent \emph{Step 1:} $\limsup_{k \to \infty} h_k^{-1}(y) \leq h^{-1}(y)$.
    \newline
    Fix an arbitrary $\epsilon > 0$. By the definition of $h^{-1}(y)$, we have $h(x) \geq y$. By right-continuity, and monotonicity, there is $x_{+} \in [h^{-1}(y), h^{-1}(y) + \epsilon)$ that is a continuity point of $h$ and satisfies $h(x_{+}) > y$. Indeed, such $x_{+}$ must exist because the assumption that $y$ is a continuity point of $h^{-1}$ ensures that $h$ cannot be constant at $y$ on an open right neighborhood of $h^{-1}(y)$. Then $h_k(x_{+}) \to h(x_{+}) > y$, so for all large enough $k$, we must have $h_k(x_{+}) \geq y$. Hence, for all sufficiently large $k$,
    \begin{equation*}
        h_k^{-1}(y) \leq x_{+} < h^{-1}(y) + \epsilon.
    \end{equation*}
    It follows that $\limsup_{k \to \infty} h_k^{-1}(y) \leq h^{-1}(y) + \epsilon$. Since this holds for any $\epsilon > 0$, we obtain
    \begin{equation*}
        \limsup_{k \to \infty} h_k^{-1}(y) \leq h^{-1}(y).
    \end{equation*}

    \bigskip

    \noindent \emph{Step 2:} $\liminf_{k \to \infty} h_k^{-1}(y) \geq h^{-1}(y)$.
    \newline
    Fix an arbitrary $\epsilon > 0$. Since $y$ is a continuity point of $h^{-1}$ and $h^{-1}(y) > 0$, there must be $x_{-} \in (h^{-1}(y) - \epsilon, h^{-1}(y))$ with $h(x_{-}) < y$. Then $h_k(x) \to h(x) < y$ as $k \to \infty$. Hence, for large enough $k$, we must have $h_k(x) < y$, which implies 
    \begin{equation*}
        h_k^{-1}(y) > x > h^{-1}(y) - \epsilon.
    \end{equation*}
    Hence, $\liminf_{k \to \infty} h_k^{-1}(y) \geq h^{-1}(y) - \epsilon$. Since this holds for any $\epsilon > 0$, it follows that
    \begin{equation*}
        \liminf_{k \to \infty} h_k^{-1}(y) \geq h^{-1}(y).
    \end{equation*}
\end{proof}

Next, we formally state the Schauder fixed point theorem and use it to show that a continuous operator mapping the space of probability measures on a compact interval to itself has a fixed point.

\begin{lemma}[\cite{schauder1930fixpunktsatz}, \cite{shapiro2016fixed} Theorem 7.1.]\label{lemma: schauder}
    Let $K$ be a nonempty compact convex subset of a normed linear space and let $f: K \to K$ be a continuous function. Then $f$ has a fixed point.
\end{lemma}

Before we proceed, we introduce the following notation and intermediate result. Let $\mathcal M$ denote the space of finite signed Borel measures on $[0,M]$. Let $\mathrm{BL}([0,M])$ be the space of bounded Lipschitz functions $f: [0,M] \to\mathbb R$ endowed with the norm
\begin{equation*}
    \|f\|_{\mathrm{BL}}:=\|f\|_\infty+\mathrm{Lip}(f),   \qquad
\mathrm{Lip}(f):=\sup_{\substack{x,y\in [0,M]\\ x\neq y}}\frac{|f(x)-f(y)|}{|x - y|}. 
\end{equation*}
Then $BL([0,M])$ equipped with the bounded Lipschitz norm is a Banach space \citep{hille2009embedding}. Let $\mathcal{M}([0,M])$ be the space of finite signed measures on $[0,M]$, which we will often denote simply by $\mathcal{M}$. Define the linear map $J:\mathcal M([0,M])\to \mathrm{BL}([0,M])^*$ by
\begin{equation*}
    J[\mu](f):=\int f\,d\mu,\qquad f\in \mathrm{BL}([0,M]).    
\end{equation*}
For any $\nu \in \mathcal{M}$, define $\lVert \nu \rVert_{BL}$ to the norm on $\mathcal{M}$ induced by the operator norm:
\begin{equation*}
    \lVert \nu \rVert_{BL} = \lVert J[\nu] \rVert_{BL([0,M])^{\ast}} = \sup_{\lVert f \rVert_{BL} \leq 1} \left | \int f d \nu \right |
\end{equation*}
and the corresponding metric $d_{BL}$ as
\begin{equation*}
    d_{BL}(\nu, \nu') = \lVert \nu - \nu' \rVert_{BL}.
\end{equation*}

\begin{lemma}\label{lemma: injectivity embedding}
    $J$ is an injective isometric embedding of $\mathcal{M}$ to the dual $BL([0,M])^{\ast}$. 
\end{lemma}

\begin{proof}
Note that the isometric property follows by construction of the norm. Next, we show $J$ is injective, that is, if $J[\mu]=0$ in $\mathrm{BL}([0,M])^*$, then $\mu=0$ as a signed measure.
Assume $J[\mu]=0$. We claim that this implies
\begin{equation*}
    \int g\,d\mu=0 \qquad \forall g\in C([0,M]).    
\end{equation*}
Fix any $g\in C([0,M])$ and $\varepsilon>0$. Since $[0,M]$ is compact, $g$ is uniformly continuous. Hence there exists $\delta>0$ such that
\begin{equation*}
    |x-y|<\delta \ \implies |g(x)-g(y)|<\varepsilon.    
\end{equation*}
Choose a partition $0=x_0<x_1<\dots<x_m=M$ with mesh $\max_i (x_i-x_{i-1})<\delta$ and define $f_\varepsilon$ as the piecewise-linear interpolant of $g$ on this partition:
\begin{equation*}
    f_\varepsilon(x_i):=g(x_i)\quad\text{for }i=0,\dots,m,
\quad \text{and } f_\varepsilon \text{ linear on each } [x_{i-1},x_i].    
\end{equation*}
Then $f_\varepsilon$ is Lipschitz, hence $f_\varepsilon\in \mathrm{BL}([0,M])$, and, by construction and uniform continuity of $g$,
\begin{equation*}
    \lVert f_\varepsilon-g \rVert_\infty \le \varepsilon.    
\end{equation*}
Because $f_\varepsilon \in \mathrm{BL}([0,M])$ and $J[\mu]=0$, we have $\int f_\varepsilon\,d\mu=0$. Therefore
\begin{equation*}
    \left|\int g\,d\mu\right|
=\left|\int (g-f_\varepsilon)\,d\mu\right|
\le \|g-f_\varepsilon\|_\infty\,\|\mu\|_{\mathrm{TV}}
\le \varepsilon\,\|\mu\|_{\mathrm{TV}},    
\end{equation*}
where $\|\mu\|_{\mathrm{TV}}$ is the total variation norm of $\mu$ (finite since $\mu$ is a finite signed measure). Since $\varepsilon>0$ was arbitrary, it follows that $\int g\,d\mu=0$ for all $g\in C([0,M])$.

Finally, by the Riesz-Markov representation theorem, a finite signed Borel measure on a compact Hausdorff space is uniquely determined by its integrals against $C([0,M])$ \citep{aliprantis2006infinite}. Hence $\int g\,d\mu=0$ for all $g\in C([0,M])$ implies $\mu=0$.
\end{proof}

\begin{lemma}\label{lemma: compact}
    Let $\mathcal{E} = J(\mathcal{M}_1)$ be the image of $\mathcal{M}_1$ under $J$. Then $\mathcal{E}$ is compact in the norm topology.
\end{lemma}

\begin{proof}
    Since $BL([0,M])^*$ endowed with the operator norm is a metric space, it suffices to prove sequential compactness of $\mathcal E$. By the isometry of $J$ (Lemma \ref{lemma: injectivity embedding}), $\mathcal E$ is (sequentially) compact if and only if $(\mathcal P,d_{\mathrm{BL}})$ is (sequentially) compact.

    Let $\{\nu_n\}_n \subset \mathcal{M}_1$ be an arbitrary sequence. Since $[0,M]$ is compact, $\{\nu_n\}_n$ is tight. By Prokhorov's theorem, there exist a subsequence $\{\nu_{n_k}\}_{k\in\mathbb N}$ and $\nu\in\mathcal P$ such that $\nu_{n_k}$ converges to $\nu$ weakly. Since weak convergence is metrized by the bounded Lipschitz metric \cite{dudley2018real}, by isometry,
    \begin{equation*}
    \lVert J(\mu_{n_k})-J(\mu) \rVert_{BL([0,M])^{\ast}}
    = d_{BL}(\mu_{n_k},\mu)\to 0.
    \end{equation*}
    Thus every sequence in $\mathcal E$ has a norm-convergent subsequence with limit in $\mathcal E$, so $\mathcal E$ is sequentially compact. Because the norm topology is metrizable, sequential compactness implies compactness.
\end{proof}

\begin{lemma}[Existence of a fixed point]\label{lemma: fixed point}
    Let $M > 0$ and $n \in \mathbb{N}$, and let $\mathcal{M}_1$ be the set of Borel probability measures on $[0,M]$, which we endow with the topology of weak convergence. Suppose $T: \mathcal{M}_1^n \to \mathcal{M}_1^n$ is continuous. Then $T$ has a fixed point.
\end{lemma}

\begin{proof} 
    We embed $\mathcal{M}_1$ into the Banach space $BL([0,M])^{\ast}$ using the injective isometry $J$ by Lemma \ref{lemma: injectivity embedding}. By Lemma \ref{lemma: compact}, the image is a compact subset of $BL([0,M])^{\ast}$. It is also straightforward to see it is convex by convexity of $\mathcal{M}_1$ and linearity of $J[\nu]$.

    Let $J^{(n)}:\mathcal{M}_1^n\to \mathcal{E}^n$ be the coordinate-wise map
    \begin{equation*}
    J^{(n)}(\nu_1,\dots,\nu_n)\equiv (J\nu_1,\dots,J\nu_n).
    \end{equation*}
    Because $J$ is a homeomorphism $\mathcal{M}_1\to\mathcal{E}$ (as it is an isometry), the product map $J^{(n)}$ is a homeomorphism $\mathcal{M}_1^n\to\mathcal{E}^n$ with the product topologies. Define
    \begin{equation*}
    f \equiv J^{(n)} \circ T \circ \bigl(J^{(n)}\bigr)^{-1} : \mathcal{E}^n \to \mathcal{E}^n.
    \end{equation*}
    Since $T$ is continuous and $J^{(n)}$ is a homeomorphism, $f$ is continuous with respect to the norm topology on $\mathcal{E}^n$.

    By Lemma \ref{lemma: fixed point}, the continuous map $f:\mathcal{E}^n\to\mathcal{E}^n$ has a fixed point $e^\ast\in \mathcal{E}^n$. Then $\nu^\ast \equiv \bigl(J^{(n)}\bigr)^{-1}(e^\ast)\in \mathcal{M}_1^n$ is a fixed point of $T$.
\end{proof}

\section{Single-item Auctions with Negative Payments}\label{sec:eso_futo}

Throughout the paper we require that payments are non-negative: $\pbic_i(\bv) \geq 0$, i.e., transfers are only allowed from the bidders to the auctioneer. If we drop this restriction but still enforce interim IR., Es\"{o} and Fut\'{o} \cite{eso1999auction} show that it is possible to construct an auction where the variance of the revenue is zero whenever the values are independent. We provide below this result and the proof for completeness.

\begin{theorem}[\cite{eso1999auction}]\label{thm:negative_payments} For any product distribution over the buyers' private valuations, and for any implementable allocation function $\babic(\bv)$, there exists a payment rule $\bpbic(\bv)$ with possibly negative payments, that is interim IR, BIC and the variance of the total revenue $R = \sum_i \pbic_i(\bv)$ is zero.
\end{theorem}

\begin{proof}
For each bidder $i$, let $z_i(v_i)$ be the interim payments in equation \eqref{eq:ret}. Also define $\bar z_i = \E[z_i(v_i)]$ as the average payment of bidder $i$. Now, define:
$$\textstyle  \pbic_i(\bv) = z_i(v_i) + \frac{1}{n-1} \sum_{j \neq i} (\bar z_j - z_j(v_j))$$
We first observe that $\E[\pbic_i(\bv) \mid v_i] = z_i(v_i)$, so the payment rule implements the allocation rule $a(v)$. Now, observe that:
$$\textstyle \sum_i \pbic_i(\bv) = \sum_i (z_i(\bv) + \bar z_i - z_i(\bv)) = \sum_i \bar z_i$$
Hence the total revenue is constant across all vectors of types.
\end{proof}

\section{Proof of Theorem \ref{thm:general_n_k}}\label{sec: appendix proof multi-item interim IR}

\begin{lemma}\label{lemma: derivative interim transfer}
    For any i.i.d. environment with $2 \leq k < n$, the interim payment rule under the efficient allocation must satisfy
    \begin{equation*}
        z'(v) = v \binom{n-1}{k-1} (n-k) F(v)^{n-k-1} (1-F(v))^{k-1} f(v).
    \end{equation*}
\end{lemma}

\begin{proof}
    In the efficient allocation, the interim allocation probability $x(v)$ is the probability that a  bidder's value is among the top $k$, meaning at least $n-k$ of the remaining $n-1$ bidders have values below $v$. This is the upper tail of a Binomial distribution with success probability $F(v)$:
$$ x(v) = \sum_{j=n-k}^{n-1} \binom{n-1}{j} F(v)^j (1-F(v))^{n-1-j} $$
Using the probabilistic identity that relates the Binomial cumulative mass function to the regularized incomplete Beta function \cite{abramowitz1965handbook} we get:
$$ x(v) = (n-k) \binom{n-1}{n-k} \int_0^{F(v)} t^{n-k-1} (1-t)^{k-1} dt $$
which has marginal:
$$ x'(v) = (n-k) \binom{n-1}{k-1} F(v)^{n-k-1} (1-F(v))^{k-1} f(v) $$
which corresponds to the probability density that the bidder is tied for the $k$-th highest value.

Together with Myerson's formula $z'(v) = v x'(v)$, this yields
\begin{equation*}
    z'(v) = v x'(v) = v \binom{n-1}{k-1} (n-k) F(v)^{n-k-1} (1-F(v))^{k-1} f(v).
\end{equation*}
\end{proof}

\begin{lemma}\label{lemma: psi 0}
    Let $\psi$ be as in \eqref{eq: solution ODE}. Then $\lim_{v \downarrow 0} \psi(v) = 0$.
\end{lemma}

\begin{proof}
    Since $N(0) = D(0) = 0$, we evaluate the limit using L'Hôpital's rule: $\lim_{v \downarrow 0} \psi(v) = \lim_{v \downarrow 0} \frac{N'(v)}{D'(v)}$.

    To compute the ratio of the derivatives, we first expand the numerator derivative $N'(v) = [z'(v) - \bar{R} G(v)](1 - F(v))$. Substituting the expression from Lemma \ref{lemma: derivative interim transfer} and $G(v)$ into the expression for $N'(v)$ and factoring out the common terms, we get:
$$ \begin{aligned} & N'(v) = \binom{n-1}{k-1} F(v)^{n-k-1} (1-F(v))^{k-1} f(v) [v(n-k)(1-F(v)) - \bar R F(v)] \\
& \begin{aligned} D'(v) &= \binom{n-1}{k-1} F(v)^{n-k-1} (1-F(v))^{k-1} f(v) [n-k - n F(v)] \end{aligned} \end{aligned}$$
For any $v > 0$, the common combinatorial coefficients and the density terms $F(v)^{n-k-1} (1-F(v))^{k-1} f(v)$ are non-zero and cancel out exactly in the ratio:
$$ \lim_{v \downarrow 0} \frac{N'(v)}{D'(v)} = \lim_{v \downarrow 0} \frac{v(n-k)(1-F(v)) - \bar R F(v)}{n-k - n F(v)} = \frac{0}{n-k} = 0.$$
\end{proof}

\section{Variance-Minimization in Multi-Item Settings with Ex-Post IR}\label{sec: appendix multi-item ex-post IR}

\subsection{Winner-Pays-Bid Is Not Optimal}\label{sec:wpb_not_optimal}

Consider $3$ bidders with unit-demand and i.i.d. uniform  valuations over $[0,1]$ and $2$ items for sale. The efficient allocation allocates to all but the lowest value agent leading to the following interim allocation rule: $x_i(v_i) = 1-(1-v_i)^2$ and corresponding interim payment rule given by the Myerson integral:
$z_i(v_i) =  \int_0^{v_i} x_i(v_i) - x_i(u) du = v_i^2 - \frac{2}{3} v_i^3$.

The first price auction, asks each bidder to submit a bid, selects the top two bidders and they pay their bids. In that auction, we can compute the bidding strategies using the Revenue Equivalence Theorem: $b_i(v_i) = z_i(v_i) / x_i(v_i)$.\footnote{Equivalently, we can consider the BIC implementation of those auctions where each bidder reports their value $v_i$ and they pay $b_i(v_i)$ if they are among the top two values.}

Hence, the second moment of the revenue is:
$$\E\left[\left(\sum_i P_i(\bv)\right)^2\right] = 6 \int_0^{1} \int_0^{v_1} \int_0^{v_2} (b(v_1) + b(v_2))^2 dv_3 dv_2 dv_1 \approx 0.26431$$

Consider the following more general payment rule, described as its BIC implementation. Each bidder reports their value $v_i$, the highest two values win the item, the highest bidder pays $b_1(v)$ and the second highest bidder pays  $b_2(v)$, where $v$ are the values respectively. This auction is BIC if it satisfies the Revenue Equivalence Theorem:
$$z_i(v_i) = v_i^2 b_1(v_i) + 2 v_i (1-v_i) b_2(v_i)$$

We can engineer functions $b_1(v_i)$ and $b_2(v_i)$ that make $b_1$ as large as possible and $b_2$ as small as possible without violating ex-post IR:

$$b_1(v_i) = \left\{\begin{aligned}
    & v_i, & & v_i \leq 0.6 \\
    & z(v_i)/v_i^2, & & v_i \geq 0.6
\end{aligned}\right. \qquad 
b_2(v_i) = \left\{\begin{aligned}
    & (z_i(v_i) - v_i^3) / (2v_i(1-v_i)), & & v_i \leq 0.6 \\
    & 0, & & v_i \geq 0.6
\end{aligned}\right.$$
It is straightforward to check that $b_1(v) \leq v$, $b_2(v) \leq v$ (hence it is ex-post IR) and that it satisfies the Revenue Equivalence Theorem. When we compute the second moment we obtain:

$$\E\left[\left(\sum_i P_i(\bv)\right)^2\right] = 6 \int_0^{1} \int_0^{v_1} \int_0^{v_2} (b_1(v_1) + b_2(v_2))^2 dv_3 dv_2 dv_1 \approx 0.2613$$

Observe that if the second highest bid is above $0.6$, the entire cost is being paid by the top bidder, therefore eliminating the cross terms in the second moment. In the remaining cases, we shift most of the payment to the highest bidder, also decreasing the effect of the cross terms. In this setting, it is not possible to completely shift all the payment to the highest bidder completely eliminating the cross terms.

\subsection{Variance-Minimizing Mechanism}\label{sec:ex_post_multi_charact}

While we are unable to characterize the optimal mechanism analytically even in this simple setting with three bidders and two items, we characterize its key properties and show how it can be computed numerically as the solution to a quadratic program. For expositional simplicity, we state the lemma below for three bidders and two items, although the analysis extends to arbitrary numbers of bidders and items. 

\begin{lemma}
\label{lemma:order_statistics}
Consider an i.i.d. environment with $k = 2$ and $n = 3$. Given any payment rule satisfying RET and ex-post IR, there is a mechanism with smaller or equal variance (or any convex risk measure of revenue) where the payments depend only on the highest and second-highest valuations, $v_{(1)}$ and $v_{(2)}$. Specifically, the payment of the highest bidder is a function $P_1(v_{(1)}, v_{(2)})$ and the payment of the second-highest bidder is a function $P_2(v_{(1)}, v_{(2)})$. The payments do not depend on the identity of the bidders nor on the valuation of the losing bidder $v_{(3)}$.
\end{lemma}

\begin{proof}
Let $\pi(\bv)$ be an arbitrary payment rule for the mechanism, where $\pi_i(\bv)$ is the payment of bidder $i$ given valuation profile $\bv = (v_1, v_2, v_3)$.

Consider the projection of this payment rule onto the order statistics. Define the new payment functions $P_1$ and $P_2$ as the expected payment of the winner with rank 1 and rank 2, conditional on the realization of $v_{(1)}$ and $v_{(2)}$:
\begin{align*}
P_1(y_1, y_2) &= \mathbb{E}_{\bv} \left[ \sum_{i \in \text{Winners}} \pi_i(\bv) \cdot \mathbb{I}(v_i = y_1) \mid v_{(1)}=y_1, v_{(2)}=y_2 \right] \\
P_2(y_1, y_2) &= \mathbb{E}_{\bv} \left[ \sum_{i \in \text{Winners}} \pi_i(\bv) \cdot \mathbb{I}(v_i = y_2) \mid v_{(1)}=y_1, v_{(2)}=y_2 \right]
\end{align*}
Note that this expectation integrates out the valuation of the losing bidder, $v_{(3)}$, and averages over the identities of the bidders.

First we observe that it preserves the expected payment of each type and hence preserves RET and as a consequence BIC. This happens because the 
interim expected payment of a bidder with type $v$ depends only on the expected payment conditioned on their rank. Since $P_1$ and $P_2$ preserve the conditional expectations by definition, the new rule $P(\cdot)$ generates the exact same interim expected payments $z(v)$ as the original rule $\pi(\cdot)$.

Finally, the variance is reduced by Jensen's inequality.
Let $R(\bv) = \sum_i \pi_i(\bv)$ be the total revenue of the original mechanism. The total revenue of the new mechanism given order statistics $(v_{(1)}, v_{(2)})$ is $\bar{R}(v_{(1)}, v_{(2)}) = P_1(v_{(1)}, v_{(2)}) + P_2(v_{(1)}, v_{(2)})$.
By construction: $\bar{R}(v_{(1)}, v_{(2)}) = \mathbb{E}[R(\bv) \mid v_{(1)}, v_{(2)}]$.
The objective is to minimize $\mathbb{E}[R(\bv)^2]$. By Jensen's Inequality:
$$
\mathbb{E}[R(\bv)^2] = \mathbb{E}_{v_{(1)}, v_{(2)}} \left[ \mathbb{E}[R(\bv)^2 \mid v_{(1)}, v_{(2)}] \right] \ge \mathbb{E}_{v_{(1)}, v_{(2)}} \left[ (\mathbb{E}[R(\bv) \mid v_{(1)}, v_{(2)}])^2 \right] = \mathbb{E}[\bar{R}^2].
$$
\end{proof}

Using  Lemma \ref{lemma:order_statistics}, we formulate the variance minimization problem as the convex optimization program over the functions $P_1(v_1, v_2)$ and $P_2(v_1, v_2)$ defined on the triangular domain $\mathcal{T} = \{(v_1, v_2) \in [0,1]^2 : v_1 \ge v_2\}$. If the bidders have i.i.d. valuations with PDF $f(v)$ and CDF $F(v)$, then we can write the problem as:

$$\begin{aligned}
& \min_{P_1, P_2} \quad \iint_{\mathcal{T}} \left( P_1(v_1, v_2) + P_2(v_1, v_2) \right)^2 6 f(v_1) f(v_2) F(v_2) \, dv_2 \, dv_1 \\
&   \begin{aligned} & \text{s.t. } &
& \int_{0}^{v} P_1(v, u) 2 f(u) F(u) du + \int_{v}^{1} P_2(u, v) 2 F(v) f(u) du = z(v) \quad \forall v \in [0,1] \\
& & & 0 \leq P_1(v_1, v_2) \leq v_1, 0 \leq P_2(v_1, v_2) \leq v_2 
\end{aligned} \end{aligned}
$$
where the first set of constraints correspond to RET and the second to set to ex-post IR and non-negativity of payments.  While difficult to solve analytically, we discretize the uniform distribution and solve a discrete version of this program, which obtains $\E[\Rev^2] \approx 0.255$. The results are displayed in the figures below.

\begin{figure}[h]
\begin{center}
\includegraphics[scale=.25]{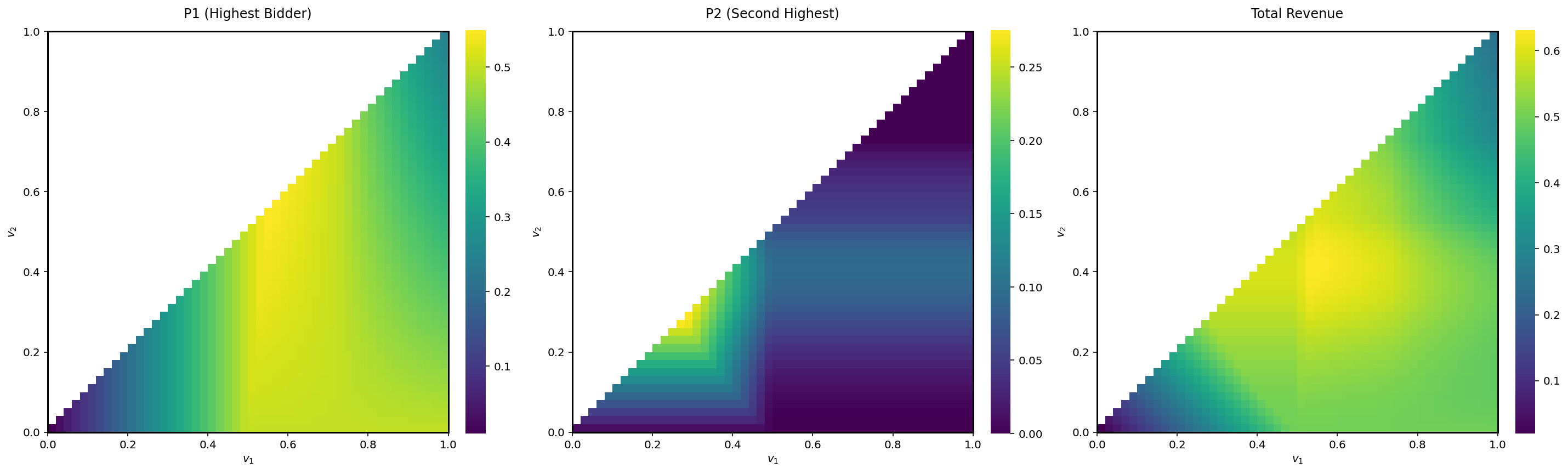}
\end{center}
\caption{Optimal payment rule}
\label{fig:opt-expost-ir}
\end{figure}

\begin{figure}[h]
\begin{center}
\includegraphics[scale=.25]{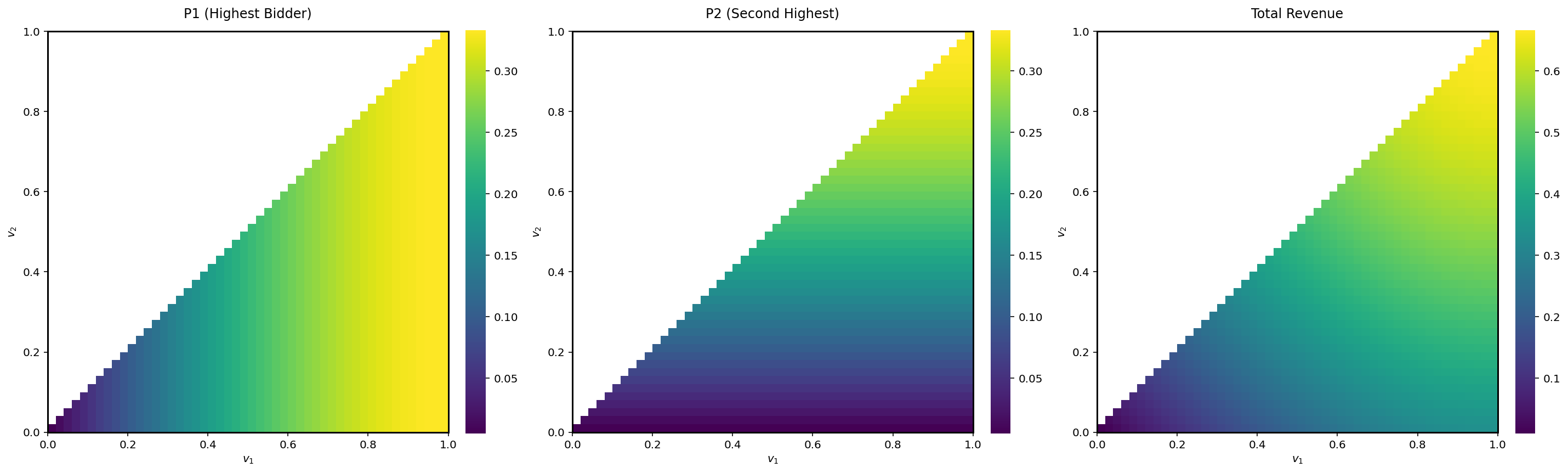}
\end{center}
\caption{First price payment rule}
\label{fig:first_price}
\end{figure}

\begin{figure}[h]
\begin{center}
\includegraphics[scale=.25]{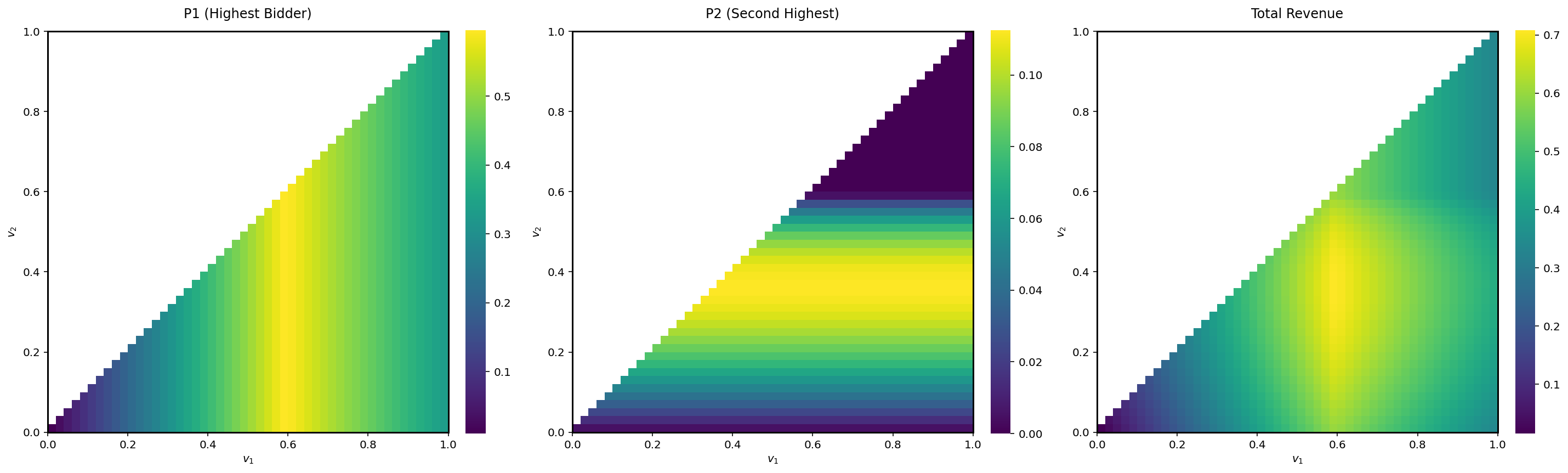}
\end{center}
\caption{Custom bidding function in Section \ref{sec:wpb_not_optimal}}
\label{fig:custom_bid}
\end{figure}

\end{document}